\documentclass[a4paper,11pt]{scrartcl}
\usepackage[pdfpagelabels]{hyperref}
\usepackage[utf8]{inputenc}
\usepackage{amsmath,amsfonts,amssymb,amsthm,mathtools}
\usepackage{amsmath,verbatim}
\usepackage{amsfonts}
\usepackage{amssymb}
\usepackage{dsfont}
\usepackage{mathrsfs}
\usepackage[T1]{fontenc}
\usepackage{lmodern}
\usepackage[ngerman,english]{babel}
\usepackage[inline]{enumitem}
\usepackage{latexsym,graphicx}
\usepackage[all]{xy}
\usepackage{mathtools}
\usepackage{color}
\usepackage{authblk}
\usepackage{amscd}
\usepackage{microtype}
\usepackage{mathrsfs,cite}
\usepackage{upgreek}
\numberwithin{equation}{section}
\usepackage{txfonts}

\DeclareMathOperator{\sgn}{sgn}
\DeclareMathOperator{\tr}{Tr}
\DeclareMathOperator{\trs}{tr}

\newtheorem{theorem}{Theorem}
\newtheorem{proposition}{Proposition}
\newtheorem{lemma}{Lemma}[section]

\theoremstyle{definition}
\newtheorem{remark}{Remark}[section]

\newtheorem{definition}{Definition}[section]

\newcommand{\dda}{\mathrm{d}}
\newcommand{\de}{\,\dda}

\renewcommand{\Re}{\textrm{Re}}
\renewcommand{\Im}{\textrm{Im}}

\newcommand{\cN}{\mathcal{N}}
\newcommand{\Tr}{\mathrm{Tr}}
\newcommand{\gesssim}{\gtrsim}

\renewcommand{\Re}{\textrm{Re}}
\renewcommand{\Im}{\textrm{Im}}
\renewcommand{\rho}{\varrho}
\newcommand{\eps}{\varepsilon}
\renewcommand{\rho}{\varrho}
\renewcommand{\epsilon}{\varepsilon}
\newcommand{\nn}{\nonumber}

\newcommand{\R} {\mathbb{R}}

\newcommand{\ou}{%
\mathrel{%
\vcenter{\offinterlineskip
\ialign{##\cr$\lesssim$\cr\noalign{\kern-1.5pt}$\gtrsim$\cr}%
}%
}%
}

\theoremstyle{definition}

\newcommand{\beq}{\begin{equation}}
\newcommand{\eeq}{\end{equation}}
\usepackage{pdfsync}

\begin{document}

\title{A note on spontaneous symmetry breaking in the mean-field Bose gas}

\author{Andreas Deuchert, Phan Th\`anh Nam and Marcin Napiórkowski}

\date{\today}

\maketitle

\begin{abstract} 
We consider the homogeneous Bose gas in the three-dimensional unit torus, where $N$ particles interact via a two-body potential of the form $N^{-1} v(x)$. The system is studied at inverse temperatures of order $N^{-2/3}$, which corresponds to the temperature scale of the Bose--Einstein condensation phase transition. We show that spontaneous $U(1)$
symmetry breaking occurs if and only if the system exhibits Bose–Einstein condensation in the sense that the one-particle density matrix of the Gibbs state has a macroscopic eigenvalue.
\end{abstract}

\setcounter{tocdepth}{2}
\tableofcontents

\section{Introduction and main results} \label{sec:intro}


\subsection{Spontaneous symmetry breaking and quasi-averages}

Spontaneous symmetry breaking plays a pivotal role in understanding complex phenomena across various domains of physics, from condensed matter systems to high-energy particle interactions. Generally,  it refers to a situation when a symmetry of the Hamiltonian or Lagrangian of a system is not present in the state under consideration - usually  a ground state or a thermal equilibrium state.

In this note, we shall investigate symmetry breaking in a system of weakly interacting bosons. In  non-relativistic quantum many-body systems, the associated symmetry is the $U(1)$ symmetry corresponding to particle number conservation of the underlying many-body Hamiltonian (cf. \eqref{eq:Hamiltonian}). It is well-known in statistical mechanics that many phase transitions are accompanied by symmetry breaking. As has been shown in the seminal papers of Bose and Einstein \cite{Bose1924,Einstein1924,Einstein1925}, the ideal (non-interacting) Bose gas undergoes the Bose--Einstein Condensation (BEC) phase transition, characterized by a macroscopic occupation of a single quantum state (cf. Sections \ref{sec:1pdm} and \ref{sec:idealBoseGas} for a brief reminder). For interacting systems, BEC has been observed in alkali gases in experiments led by Ketterle \cite{Ketterle-95} and by Cornell and Wieman \cite{CorWie-95}, but its theoretical understanding  remains challenging. Mathematically, it has been shown to occur for simplified models, in particular in the so-called mean-field scaling. 
The main result of this paper shows that the BEC phase transition in the mean-field Bose gas is indeed equivalent to the breaking of the $U(1)$ symmetry of the underlying many-body system.

While the definition of Bose-Einstein condensation (cf. \eqref{eq:definitionBEC}) is rather simple, a 
rigorous description of symmetry breaking is a bit more subtle. In order to explain it, let us focus on a translation-invariant system. In this case, the macroscopically occupied single particle quantum state in the condensed phase corresponds to the zero momentum mode, i.e. the expected number of particles with zero momentum   $\langle a^*_0 a_0 \rangle$ (cf. \eqref{eq:annihilationOperator} and \eqref{eq:creationOperator} for a definition of creation and annihilation operators) becomes proportional to the total number of particles in the system. This makes the zero momentum mode the most relevant one and, as suggested by   Bogoliubov \cite{Bogoliubov-47}, justifies to treat them classically by replacing the operators $a_0$ and $a_0^*$ by $c$-numbers and thus equating $\langle a^*_0 a_0 \rangle$ and $|\langle a_0 \rangle|^2$. As the latter quantity cannot be non-zero in any state that has a fixed number of particles (or more generally in any state that preserves the total number of particles), it serves as a marker of spontaneous symmetry breaking. 

In order to mathematically implement those ideas, Bogoliubov \cite{Bogoliubov-1961,Bogoliubov-1970} designed a limiting procedure that nowadays goes under the name of \textit{Bogoliubov's quasi averages} and has found applications in various fields of physics. It consists of the following two steps. First, in  order to make the expectation value of $a_0$ possibly non-zero, one adds to the   many-body Hamiltonian a perturbation $\lambda \sqrt{V}(a_0^*+a_0)$ that breaks particle number conservation (here $V$ is the volume of the system). Let us denote by $\langle \cdot \rangle_\lambda$ the expectation value in such a perturbed state.  One then considers the double limit 
$$\lim_{\lambda\to 0} \lim_{V\to \infty} \frac{|\langle a_0 \rangle_\lambda|^2}{V}$$
and says that a non-zero value means that the symmetry is broken in the system. We stress the order of limits, which means that one first takes the macroscopic limit and only later 'removes' the symmetry breaking perturbation. Here, we also assume the macroscopic limit to be the \textit{thermodynamic limit}, as originally considered by Bogoliubov.  

Note that  the Cauchy-Schwarz inequality implies that 
$$ \frac{|\langle a_0 \rangle_\lambda|^2}{V}\leq \frac{\langle a^*_0 a_0 \rangle_\lambda}{V}$$
and thus spontaneous symmetry breaking implies BEC in the sense of a macrosopic occupation of the zero momentum mode. Note, however, the dependence of the right-hand side on $\lambda$. In 2005  Lieb, Seiringer and Yngvason in \cite{LieSeiYng-05}, and, independently, S\"{u}t\H{o} in \cite{Suto-05} proved that in the limit $V \to \infty$ there is equality in the above equation for almost every $\lambda$.
Those works were inspired by the result of Ginibre \cite{Ginibre-1968}, which gave the first rigorous justification of the $c$-number substitution. In \cite{LieSeiYng-05} the authors have also shown that if there is BEC in the usual sense, then there is spontaneous symmetry breaking, i.e.
\begin{equation} \label{eq:BECimpliesSSB}
\lim_{V\to \infty} \frac{\langle a^*_0 a_0 \rangle_{\lambda=0}}{V}\leq \lim_{\lambda\to 0} \lim_{V\to \infty} \frac{|\langle a_0 \rangle_\lambda|^2}{V}. 
\end{equation}
A different approach that uses operator algebraic tools and correlation inequalities can be found in \cite{FannesPuleVerbeure1982}. See also \cite{PuleVerbeureZagrebnov2005} for a generalization to nonhomogeneous systems. Note that all these statements are \textit{conditional} in the sense that proving BEC \textit{in the thermodynamic limit} remains a major open problem in mathematical physics.

While the problem in the thermodynamic limit remains open, one can ask the same question in other scaling regimes that are easier to handle. In the present paper, we are interested in the mean-field scaling, where $N$ particles in the unit torus interact via an interaction potential of the form $N^{-1}v(x - y)$. This model corresponds to the physical situation of a weakly interacting gas, where the interaction energy per particle is proportional to the spectral gap of the one-body kinetic operator. The zero-temperature properties of this model are rather well understood: the ground state exhibits BEC \cite{LewNamRou-14,NamNap-21,BosPetSei-21}, while the excited states are well described by Bogoliubov theory \cite{Seiringer-11,GreSei-13,DerNap-13,LewNamSerSol-15,NamSei-15}. We will focus on systems at positive temperatures that are comparable to the critical temperature of the BEC phase transition. In this case, an exhaustive description of the Gibbs state of the (translation-invariant) mean-field Bose gas has been given in \cite{DeuNamNap-25}; see also \cite[Section~1.6]{DeuSei-21} for a discussion of the accuracy of Hartree theory in the trapped system. However, symmetry breaking in the spirit of (1.1) has not yet been explained in great detail, and fulfilling this task is the main purpose of our paper.

The remaining part of the introduction will be devoted to a detailed description of the model and a precise formulation of the main results.

\subsubsection*{Notation} 
We write $a \lesssim b$ to say that there exists a constant $C>0$ independent of the relevant parameters (for instance, the particle number or  the inverse temperature) such that $a \leq C b$ holds. If $a \lesssim b$ and $b \lesssim a$ we write $a \sim b$, and $a \simeq b$ means that the leading orders of $a$ and $b$ are equal in the limit considered. In case the constant depends on a parameter $k$, we write $a \lesssim_k b$ and $a \sim_k b$.

\subsection{Bosonic Fock space and Hamiltonian} 
We consider a system of bosonic particles confined to the three-dimensional flat torus $\Lambda=[0,1]^3$. The one-particle Hilbert space of the system is given by $L^2(\Lambda)$. We are interested in the grand canonical ensemble, that is, in systems with a fluctuating particle number. The Hilbert space of the entire system is therefore given by the bosonic Fock space over $L^2(\Lambda)$: 
\begin{equation}
	\mathscr{F}(L^2(\Lambda)) = \bigoplus_{n=0}^{\infty} L^2_{\mathrm{sym}}(\Lambda^n) = \mathbb{C} \oplus L^2(\Lambda) \oplus L^2_{\mathrm{sym}}(\Lambda^2) \oplus \cdots 
    \label{eq:FockSpace}
\end{equation}
Here, $L^2_{\mathrm{sym}}(\Lambda^n)$ denotes the closed linear subspace of $L^2(\Lambda^n)$ consisting of those functions $\Psi(x_1,...,x_n)$ that are invariant under any permutation of the particle coordinates $x_1, ..., x_n \in \Lambda$. 

On the $n$-particle Hilbert space $L^2_{\mathrm{sym}}(\Lambda^n)$ we define the Hamiltonian
\begin{equation}
    \mathcal{H}^n_\eta = \sum_{i=1}^n - \Delta_i + \frac{1}{\eta} \sum_{1 \leq i < j \leq n} v(x_i-x_j).
    \label{eq:nParticleHamiltonian}
\end{equation}
Here $\Delta_i$ denotes the Laplace operator on $L^2(\Lambda)$ with periodic boundary conditions acting on the coordinate $x_i$ and represents the kinetic energy of the particles. Moreover, $v : \Lambda \to \mathbb{R}$ describes the interaction between the particles and the coupling constant $\eta^{-1}>0$ stands for the strength of the interaction. We will assume that $v$ is bounded, and hence $\mathcal{H}^n_\eta$ is a self-adjoint operator on the domain $H^2_{\mathrm{sym}}(\Lambda^n)$ of the non-interacting Hamiltonian. The Hamiltonian of the entire system acts on a suitable dense subset of $\mathscr{F}(L^2(\Lambda))$ and is defined by
\begin{equation}
    \mathcal{H}_{\eta} = \bigoplus_{n=0}^{\infty} \mathcal{H}^n_{ \eta}. 
    \label{eq:FockSpaceHamiltonian1}
\end{equation}

To state an alternative representation of this Hamiltonian, we introduce the creation and annihilation operators $a^*_p$ and $a_p$ of a particle in the one-particle function $\varphi_p(x) = e^{\mathrm{i}p \cdot x}$ with $p \in \Lambda^* = 2 \pi \mathbb{Z}^3$. The annihilation operator $a_p$ is a map from $L^2_{\mathrm{sym}}(\Lambda^n)$ to $L^2_{\mathrm{sym}}(\Lambda^{n-1})$ for $n \geq 1$ that satisfies $a_p \Omega = 0$ for $\Omega = (1,0,0,...) \in \mathscr{F}$. Its action on an n-particle function $\psi \in L^2_{\mathrm{sym}}(\Lambda^n)$ is defined by
\begin{equation}
    (a_p \psi)(x_1,...,x_{n-1}) = \sqrt{n} \int_{\Lambda} e^{-\mathrm{i}p \cdot x} \psi(x_1,...,x_{n-1},x) \de x,
    \label{eq:annihilationOperator}
\end{equation}
and for general $\psi \in \mathscr{F}$ this action is extended by linearity. The creation operator $a^*_p$ is the adjoint of $a_p$ and acts on $\psi \in L^2_{\mathrm{sym}}(\Lambda^n)$ as
\begin{equation} \label{eq:creationOperator}
    (a^*_p \psi)(x_1,...,x_{n+1}) = \frac{1}{\sqrt{n+1}} \sum_{j=1}^n e^{\mathrm{i} p \cdot x_j} \psi(x_1,...,x_{j-1}, x_{j+1}, ..., x_n). 
\end{equation}
The above family of operators satisfies the canonical commutation relations
\begin{equation}
	[a_p,a_q^*] = \delta_{p,q}, \quad \quad [a_p,a_q] = 0 = [a^*_p,a^*_q].
	\label{eq:CCR}
\end{equation}
An alternative representation of the  Hamiltonian in \eqref{eq:FockSpaceHamiltonian1} in terms of these operators is given by
\begin{equation}
	\mathcal{H}_\eta = \sum_{p \in \Lambda^*} p^2 a^*_p a_p + \frac{1}{2 { \eta}} \sum_{p,u,v \in \Lambda^*} \hat{v}(p) a_{u+p}^* a_{v-p}^* a_u a_v 
	\label{eq:Hamiltonian}
\end{equation}
with the Fourier coefficients $\hat{v}(p) = \int_{\Lambda} v(x) e^{-\mathrm{i}p \cdot x} \de x$ of the interaction potential $v$.  
\subsection{The grand canonical ensemble}
\label{sec:freeenergy}
We are interested in a bosonic many-particle system described by the grand canonical ensemble. The usual parameters of this ensemble\footnote{A third parameter is the volume of the torus. Since there is a free parameter in our model, we decided to set it equal to one.} are the inverse temperature $\beta > 0$ and the chemical potential $\mu \in \mathbb{R}$. The choice of the latter allows one to obtain any desired value for the expected number of particles $N(\beta,
\mu)> 0$ in the system. This, in particular, allows us to ensure the interpretation of the mean-field scaling as in the canonical setup. 

The set of states on the bosonic Fock space is given by
\begin{equation}
    \mathcal{S} = \left\{ \Gamma \in \mathcal{B}(\mathscr{F}) \ | \ 0 \leq \Gamma, \tr \Gamma = 1 \right\},
	\label{eq:states}
\end{equation}
where $\mathcal{B}(\mathscr{F})$ denotes the set of bounded operators on $\mathscr{F}(L^2(\Lambda))$. For a state $\Gamma \in \mathcal{S}$ the grand potential functional $\mathcal{G}$ is defined by
\begin{equation}
    \mathcal{G}(\Gamma) = \tr[(\mathcal{H}_{\eta} - \mu \mathcal{N} ) \Gamma ] - \frac{1}{\beta} S(\Gamma) \quad \text{ with the von-Neumann entropy } \quad S(\Gamma) = - \tr[\Gamma \ln(\Gamma)].
    \label{eq:grandPotentialFunctional}
\end{equation}
Here $\mu \in \mathbb{R}$ and 
\begin{equation}
    \mathcal{N} = \bigoplus_{n=0}^{\infty} n = \sum_{p \in \Lambda^*} a^*_p a_p
    \label{eq:numberOperator}
\end{equation}
denote the chemical potential and the number operator, respectively. Its minimum equals the grand potential
\begin{equation}
    \Phi(\beta,\mu) = \min_{\Gamma \in \mathcal{S}} \mathcal{G}(\Gamma) = -\frac{1}{\beta} \ln\left( \tr \exp(-\beta(\mathcal{H}_{\eta} - \mu \mathcal{N})) \right)
    \label{eq:grantPotential}
\end{equation}
and the unique minimizer of $\mathcal{G}$ in the set $\mathcal{S}$ is the grand canonical Gibbs state
\begin{equation}
	G_{\beta,\mu} = \frac{\exp\left( -\beta (\mathcal{H}_{\eta} - \mu \mathcal{N}) \right) }{\tr \exp\left( -\beta (\mathcal{H}_{\eta} - \mu \mathcal{N}) \right) }.
	\label{eq:interactingGibbsstate}
\end{equation}
The expected number of particles in the Gibbs state $G_{\beta,\mu}$ will be denoted by $N(\beta,\mu)$, that is,
\begin{equation}\label{eq:expnumberparticlesGibbs}
    N(\beta,\mu)= \tr [\cN G_{\beta,\mu}].
\end{equation}
This number will later be shown (cf. Proposition \ref{prop:expnumberparticles}) to be proportional to $\eta$. Thus the choice of the coupling constant $\eta^{-1}$ in front of the interaction implements a mean-field scaling. 

\subsection{One-particle reduced density matrix and Bose--Einstein condensation} \label{sec:1pdm}
 Let $\Gamma \in \mathcal{S}$ be a state for which $\mathcal{N} \Gamma$ is trace class. We define the one-particle density matrix (1-pdm) $\gamma_{\Gamma} \in \mathcal{B}(L^2(\Lambda))$ of $\Gamma$ via its integral kernel in Fourier space by
 \begin{equation}
     \gamma_{\Gamma}(p,q) = \Tr[a_q^* a_p \Gamma].
     \label{eq:1pdm}
 \end{equation}
 It is a nonnegative trace-class operator, which satisfies $\trs[\gamma_{\Gamma}] = \Tr[\mathcal{N} \Gamma]$.
  
 Following \cite{PenOns-56} we say that a sequence of states $\Gamma_N \in \mathcal{S}$ with $\Tr[\mathcal{N} \Gamma_N] = N$ displays Bose--Einstein condensation (BEC) iff
\begin{equation}
    \liminf_{N \to \infty} \sup_{\Vert \psi \Vert = 1} \frac{\langle \psi, \gamma_{\Gamma_N} \psi \rangle}{N} > 0,
    \label{eq:definitionBEC}
\end{equation}
that is, iff the largest eigenvalue of $\gamma_{\Gamma_N}$ growths proportionally to $N$. The largest eigenvalue of $\gamma_{\Gamma_N}$ divided by $N$ and the corresponding eigenvector are called the condensate fraction and the condensate wave function, respectively. 

Let us mention that, apart from the results cited in the introduction (for the mean-field scaling), BEC has been shown to hold also in other scaling regimes. We refer to \cite{LieSei-02,NamRouSei-16,BocBreCenSch-18,BocBreCenSch-20,NamNapRicTri-21,AdhBreSch-21,Hainzl-21,BreBroCarOld-24,Fournais-21} for the ground state and to \cite{DeuSeiYng-19,DeuSei-20,DeuSei-21} for the Gibbs state. 

\subsection{The ideal Bose gas on the torus} \label{sec:idealBoseGas}
  If one sets $v = 0$ in the definitions of the Hamiltonian in \eqref{eq:Hamiltonian} and the Gibbs state in \eqref{eq:interactingGibbsstate} we obtain a non-interacting model called the ideal Bose gas, which is exactly solvable. 

The $1$-pdm of the Gibbs state of the ideal Bose gas is given by
\begin{equation}
    \gamma^{\mathrm{id}} = \sum_{p \in \Lambda^*} \frac{1}{\exp(\beta \left( p^2 -  \mu_0 ) \right) - 1} | \varphi_p \rangle \langle \varphi_p |,
    \label{eq:1pdmIdealGas}
\end{equation}
where $| \varphi_p \rangle \langle \varphi_p |$ denotes the orthogonal projection onto the one-particle function $\varphi_p(x) = e^{\mathrm{i} p \cdot x}$. Accordingly, the chemical potential $\mu_0(\beta,N) < 0$ can be defined as the unique solution to the equation
\begin{equation}
	N =  \trs_{L^2(\Lambda)} [ \gamma^{\mathrm{id}} ] = \sum_{p \in \Lambda^*} \frac{1}{\exp(\beta \left( p^2 -  \mu_0 ) \right) - 1},
	\label{eq:idealgase1pdmchempot}
\end{equation} 
where $N$ denotes the expected number of particles in the system.

As has been realized by Einstein in \cite{Einstein1925}, the ideal Bose gas displays a phase transition (the BEC phase transition) in the large $N$ limit: the largest eigenvalue $N_0(\beta,\mu_0)$ of $\gamma^{\mathrm{id}}$ in \eqref{eq:1pdmIdealGas} behaves as
\begin{equation}
	N_0(\beta,\mu_0(\beta,N)) = \frac{1}{\exp(-\beta \mu_0)-1} \simeq N \left[ 1 - \left( \frac{\beta_{\mathrm{c}}}{\beta} \right)^{3/2} \right]_+, \quad \text{ where } \quad \beta_{\mathrm{c}} = \beta_{\mathrm{c}}(N) = \frac{1}{4 \pi} \left( \frac{N}{\upzeta(3/2)} \right)^{-2/3}
	\label{eq:crittemp}
\end{equation}
in the limit $N \to \infty$. Here, $\upzeta$ denotes the Riemann zeta function and $[x]_+ = \max\{ 0,x \}$. The function $N_0(\beta,\mu_0)$ is also the expected number of particles in the system with momentum equal to zero. We highlight that, because of the $N$-dependence of $\beta_{\mathrm{c}}$, $\beta$ usually depends on the particle number\footnote{In a box with side length $L > 0$ the inverse critical temperature of the BEC phase transition in the ideal gas is given by $\beta'_{\mathrm{c}} = \frac{1}{4 \pi} \left( \frac{\varrho}{\upzeta(3/2)} \right)^{-2/3}$ with the density $\varrho = N/L^3$. The inverse critical temperature in \eqref{eq:crittemp} depends on $N$ because we have chosen $L=1$.}. If we choose $\beta$ so that $\lim_{N \to \infty} \beta_{\mathrm{c}}/\beta = 0$ (zero temperature limit), then $N_0(\beta,\mu_0(\beta,N))/N \to 1$. However, in general, it is possible to have $N_0(\beta,\mu_0(\beta,N))/N \to g \in [0,1]$, that is, one can have a macroscopic number of particles (order $N$) with $p = 0$ (the BEC) and a macrosopic number of thermally excited particles. Finally, \eqref{eq:crittemp} implies the following behavior: if $\beta = \kappa \beta_{\mathrm{c}}$ with $\kappa > 1$ we have $\mu_0 \simeq -(\beta N_0)^{-1} \sim -N^{-1/3}$, while $\mu_0 \sim -\beta^{-1} \sim -N^{2/3}$ holds for $\beta = \kappa \beta_{\mathrm{c}}$ with $\kappa < 1$.

The grand potential of the ideal gas is
\begin{equation}
    \Phi^{\mathrm{id}} (\beta,\mu_0) = -\frac{1}{\beta} \ln\left( \tr \exp(-\beta \de \Upsilon(-\Delta - \mu_0)) \right) = \frac{1}{\beta} \sum_{p \in \Lambda^*} \ln\left( 1 - \exp(-\beta(p^2 - \mu_0)) \right).
    \label{eq:grandPotentialIdealGas}
\end{equation}
For later reference we also introduce
\begin{equation}\
    \Phi_+^{\mathrm{id}} (\beta,\mu_0) = \frac{1}{\beta} \sum_{p \in \Lambda_+^*} \ln\left( 1 - \exp(-\beta(p^2 - \mu_0)) \right).
    \label{eq:grandPotentialIdealGasExcited}
\end{equation}

\subsection{Main results}
Recall that $N(\beta,\mu)$ defined in \eqref{eq:expnumberparticlesGibbs} is the expected number of particles in the unperturbed Gibbs state $G_{\beta,\mu}$. As discussed in the introduction, in order to describe $U(1)$  spontaneous symmetry breaking,  we define the Hamiltonian 
\begin{equation}
    \mathcal{H}^{\lambda}_{\eta} = \mathcal{H}_{\eta} + \lambda N(\beta,\mu)^{1/2} (a_0 + a_0^*)
    \label{eq:HamiltonianWithSymmetryBreakingPerturb}
\end{equation}
with a coupling parameter $\lambda\in \mathbb{R}$ in front of the perturbation. Note that we replaced the square root of the volume by that of the particle number because our volume is fixed. The Hamiltonian $\mathcal{H}_\eta$ is defined in \eqref{eq:Hamiltonian}. We restrict ourselves to the case of real $\lambda$ for the sake of simplicity and without loss of generality. For $\lambda \in \mathbb{C}$ the perturbation term needs to be replaced by $N(\beta,\mu)^{1/2} (\overline{\lambda} a_0 + \lambda a_0^*)$. 

  Our first result shows that $N(\beta,\mu)$ grows proportionally to $\eta$. Before we state it, we define an effective chemical potential $\widetilde{\mu} < 0$ given as the solution of the  equation   
\begin{equation}
        \sum_{p \in \Lambda^*} \frac{1}{e^{\beta(p^2 - \widetilde{\mu})}-1} = \frac{(\mu - \widetilde{\mu})\eta}{\hat{v}(0)},
        \label{eq:GrantCanonicalEffectiveIddealGasChemPot}
    \end{equation}
where $\beta,\eta, \hat{v}(0)>0$ and $\mu\in \mathbb{R}$. As proven in Appendix \ref{app:effectiveChemicalPotential}, there is a unique solution $\widetilde{\mu}$ in the set $(-\infty,0)$. Furthermore, under the assumption $-\eta^{2/3} \lesssim \mu \lesssim 1$ and $\beta = \kappa \beta_{\mathrm{c}}(\eta)$, in the limit $\eta \to \infty$, there exists a constant $c>0$ such that $\widetilde{\mu}$ satisfies 
    \begin{equation}
        c \leq \mu - \widetilde{\mu} \leq c^{-1}. 
        \label{eq:muMinusMu0Bound1}
    \end{equation}
As our main result in Theorem~\ref{thm:SSBCondensate} below shows, the condition on $\mu$ allows us to describe a system in the condensed and the non-condensed phases. It can also be motivated by the bounds for the chemical potential in \cite[Lemma~6.1]{DeuNamNap-25}. See also Remark~6.1 in the same reference for a discussion.

\begin{proposition}[Expected number of particles]\label{prop:expnumberparticles}
   Let the interaction potential $v : \Lambda \to \mathbb{R}$ be a periodic function, whose Fourier coefficients satisfy $0 \leq \hat{v} \in \ell^1(\Lambda^*)$ and $\hat{v}(0) > 0$. We consider the limit $\eta \to \infty$, $\beta \sim \eta^{-2/3}$.  If $-\eta^{2/3} \lesssim \mu \lesssim 1$, then 
   \begin{equation}\label{eq:particlesnumberasympt}
       N(\beta,\mu)=\frac{(\mu - \widetilde{\mu})\eta}{\hat{v}(0)}+O\big(\eta^{5/6}\sqrt{\ln\eta}\big).
   \end{equation}
\end{proposition}
\begin{remark}
 Proposition \ref{prop:expnumberparticles} explains why we call the scaling of the interaction in the Hamiltonian \eqref{eq:Hamiltonian} mean-field. Indeed, in the canonical picture the mean-field scaling corresponds to a coupling constant that is inversely proportional to the number of particles in the system.     
\end{remark}

In Proposition \ref{prop:expnumberparticles} we only assumed that $\beta \sim \eta^{-2/3}$. In order to set a reference temperature in the spirit of \eqref{eq:crittemp}, let us consider the question of condensation in the self-consistent ideal gas, i.e. the one described by \eqref{eq:GrantCanonicalEffectiveIddealGasChemPot}. To this end notice that whenever $\mu\leq 0$ there is no condensation. Indeed, it follows from \eqref{eq:muMinusMu0Bound1} that in this situation $\widetilde{\mu} \leq \mu -c \leq -c$ for some positive constant $c$. Then
\begin{equation*}
    N_0(\beta,\widetilde{\mu})\coloneqq \frac{1}{e^{-\beta \widetilde{\mu}}-1} \leq \frac{1}{e^{\beta c}-1} \leq (\beta c)^{-1}\ll \eta
\end{equation*}
and there is no condensation. The same argument applies if $0\leq \mu=o(1)$. Thus condensation in the self-consistent ideal gas can only occur when $\mu \sim 1$. To compute the inverse critical temperature, we neglect the term in the sum on the left-hand side of \eqref{eq:GrantCanonicalEffectiveIddealGasChemPot} corresponding to $p=0$ and set $\widetilde{\mu}=0$. When we also approximate the remaining sum by an integral we find
$$
    \frac{1}{(2\pi)^3} \int_{\mathbb{R}^3} \frac{1}{e^{\beta p^2}-1} \de p \simeq \frac{\mu \eta}{\hat{v}(0)}.
$$
The above equation provides a relation between $\beta$ and $\mu$ that defines the critical point. Assume that we now increase $\mu$ by a constant of order one or increase $\beta$ on the scale $\eta^{-2/3}$. In this case we have to insert back the term $N_0(\beta,\widetilde{\mu}) \sim \eta$ on the left-hand side to compensate for this change. We also have to insert back $\widetilde{\mu} < 0$. However, the leading order behavior of the integral on the left-hand side and that of the right-hand side are not affected when we insert back $\widetilde{\mu} \sim -\eta^{-1/3}$ as in \eqref{eq:GrantCanonicalEffectiveIddealGasChemPot}. We conclude that an increase of $\mu$ or $\beta$ in the described way leads to condensation. This discussion motivates the following definition.

\begin{definition}
For $\mu \in \mathbb{R}$ and $\eta, \hat{v}(0) > 0$ we define 
\begin{equation} \label{def:crtitempselfconsistent}
\beta_{\mathrm{c}}(\mu,\eta):=\begin{cases}
    \dfrac{1}{4 \pi} \left( \dfrac{\mu \, \eta}{\hat{v}(0) \upzeta(3/2)} \right)^{-2/3} \quad &\text{if}\quad \mu > 0,\\
    +\infty \quad &\text{if} \quad \mu \leq 0.
\end{cases}
\end{equation}
\end{definition}
If we assume that $\beta \sim \eta^{-2/3}$, $-\eta^{2/3} \lesssim \mu \lesssim 1$, and that the limit $\lim_{\eta \to \infty} \beta/\beta_{\mathrm{c}}(\mu,\eta) = \kappa \in [0,\infty)$ exists then we have
\begin{equation}
    \lim_{\eta \to \infty}\frac{N_0(\beta,\widetilde{\mu})}{N(\beta,\mu)}= \left[1-\frac{1}{\kappa^{3/2}}\right]_+.
    \label{eq:condensateFractionEffectiveTheory}    
\end{equation}

Let's get back to symmetry breaking. Recall the perturbed Hamiltonian \eqref{eq:HamiltonianWithSymmetryBreakingPerturb}. The corresponding Gibbs state is given by
\begin{equation}
    G_{\beta,\mu}^{\lambda} = \frac{\exp\left( -\beta (\mathcal{H}^{\lambda}_{\eta} - \mu  \mathcal{N}) \right) }{\tr \exp\left( -\beta (\mathcal{H}^{\lambda}_{\eta} - \mu \mathcal{N}) \right) }.
    \label{eq:perturbedGibbsstate}
\end{equation}
 Note also that while $[\mathcal{N},G_{\beta,\mu}] = 0$, the perturbed Gibbs state does not commute with $\mathcal{N}$ if $\lambda \neq 0$.

Our main result is captured in the following theorem.
\begin{theorem}
\label{thm:SSBCondensate}
Let the interaction potential $v : \Lambda \to \mathbb{R}$ be a periodic function, whose Fourier coefficients satisfy $0 \leq \hat{v} \in \ell^1(\Lambda^*)$ and $\hat{v}(0) > 0$. We consider the limit $ \eta \to \infty$, $\beta \sim \eta^{-2/3}$, $-\eta^{2/3} \lesssim \mu \lesssim 1$ and assume that the limit $\lim_{\eta \to \infty} \beta/\beta_{\mathrm{c}}(\mu,\eta) = \kappa \in [0,\infty)$ exists with $\beta_{\mathrm{c}}(\mu,\eta)$ in \eqref{def:crtitempselfconsistent}. Then the following holds:
\begin{enumerate}[label=(\alph*)]
\item Let $\gamma_{\beta, \mu }$ be the 1-pdm of the unperturbed Gibbs state $G_{\beta, \mu}$ in \eqref{eq:interactingGibbsstate}. We have
\begin{equation} \label{eq:thmmfbec}
    \lim_{ \eta \to \infty} \sup_{\Vert \psi \Vert = 1} \frac{\langle \psi, \gamma_{\beta, \mu} \psi \rangle}{ N(\beta,\mu)} = \lim_{ \eta \to \infty} \frac{\Tr[a_0^* a_0 G_{\beta, \mu}]}{ N(\beta,\mu)} = \left[ 1 - \frac{1}{\kappa^{3/2}} \right]_+.
\end{equation}
That is, the unperturbed system displays a BEC phase transition (here we use the definition of BEC in \eqref{eq:definitionBEC}) with the inverse critical temperature $\beta_{\mathrm{c}}(\mu,\eta)$ in \eqref{def:crtitempselfconsistent} of the self-consistent ideal gas. 
\item For the perturbed Gibbs state $G^{\lambda}_{\beta, \mu}$, we have
\begin{equation}
    \lim_{\lambda \to 0} \lim_{\eta\to \infty} \frac{|\tr[ a_0 G^{\lambda}_{\beta, \mu} ]|}{  N(\beta,\mu)^{1/2}} = \sqrt{\left[ 1 - \frac{1}{\kappa^{3/2}} \right]_+}.
    \label{eq:SSBCondensate}
\end{equation}
That is, the $U(1)$-symmetry of the system is broken in the sense that the above Bogoliubov quasi-average has a nonzero limit iff the system displays BEC in the sense of \eqref{eq:definitionBEC}.
\item We have
\begin{equation}
    \lim_{\lambda \to 0} \lim_{ \eta \to \infty} \frac{|\tr[ a_0^* a_0 G^{\lambda}_{\beta,  \mu} ]|}{  N(\beta,\mu)} = \left[ 1 - \frac{1}{\kappa^{3/2}} \right]_+.
    \label{eq:continuityCondensateFraction}
\end{equation}
That is, the limit of the condensate fraction as $ \eta \to \infty$ is continuous at $\lambda = 0$. 
\end{enumerate}
\end{theorem}

We have the following remarks concerning the above theorem.

\begin{remark}
\begin{itemize}
    \item As already mentioned below \eqref{eq:definitionBEC}, the BEC phase transition has been established for systems in more challenging scaling regimes. Point (a) in Theorem~\ref{thm:SSBCondensate} can be seen as a simpler version of parts of the results in \cite{DeuSei-21}. We state it here mainly to provide a formula for the condensate fraction in the interacting model. Note, however, that the particle number has been used as a thermodynamic variable in \cite{DeuSei-21}, while we fix the chemical potential. Point (b) is the main part of Theorem~\ref{thm:SSBCondensate}. It shows that the quasi-average of $a_0$ in the perturbed Gibbs state converges in the limit $\lim_{\lambda \to 0} \lim_{N \to \infty}$ to the square root of the condensate fraction. That is, point (b) yields an alternative definition of BEC. Finally, point (c) shows that the condensate fraction in the perturbed Gibbs state, converges, for $\lambda \to 0$, to that of the unperturbed Gibbs state, which could be used as a third definition of the condensate fraction. 
    \item One could also define BEC in translation-invariant systems via off-diagonal long range order of the integral kernel of the 1-pdm in position space. However, as has been shown in Yang \cite{Yang-1962}, this is equivalent to defining BEC via \eqref{eq:definitionBEC}.
    \item The statement in Theorem~\ref{thm:SSBCondensate} not only holds for the Gibbs state but also for approximate minimizers of the grand potential functional $\mathcal{G}$ in \eqref{eq:grandPotentialFunctional}. In fact, in order for Theorem \ref{thm:SSBCondensate} to hold for a state that is not necessarily a Gibbs state, one requires that the perturbed grand potential functional evaluated at those states satisfies the upper bounds in Section \ref{sec:boundsgrandPotential}. 
\end{itemize}
\end{remark}

\subsection{Organization of the article and strategy of proof}

The proofs of both main results, Proposition \ref{prop:expnumberparticles} and Theorem \ref{thm:SSBCondensate}, are based on a variational approach. 

In Section \ref{sec:expectednumberproof} we provide upper and lower bounds for the grand potential related to the unperturbed Hamiltonian. The proof of the upper bound relies on a trial state that treats the zero momentum mode as a coherent state and the excited particles as those of an ideal gas with an appropriately chosen chemical potential. The lower bound is a straightforward consequence of the Onsager Lemma~\ref{lem:OnsagersLemma}. The grand potential bounds are then used to prove Proposition~\ref{prop:expnumberparticles}. 

The proof of Theorem \ref{thm:SSBCondensate}  is divided into several steps. First in Section \ref{sec:boundsgrandPotential}, we provide upper and lower bounds for the grand potential of the perturbed Hamiltonian \eqref{eq:HamiltonianWithSymmetryBreakingPerturb}. In fact, apart from the the symmetry breaking perturbation, we need to consider also another perturbation of the form $\delta a_0^* a_0$. The upper bound uses a similar trial state as the one used for the non-perturbed Hamiltonian. However, now, the condensate part of the trial state has to take into account the right phase - as induced by the symmetry breaking term. 
The lower bound is more involved. While we can still use the Onsager Lemma as in the analysis of the unperturbed grand potential, the perturbed grand potential requires also a $c$-number substitution and entropic inequalities  in the spirit of \cite{DeuSei-20} (cf. Section \ref{ssec:cnumber}), leading to an effective minimization problem which can be solved approximately. 

The bounds on the perturbed grand potential are used in Section~\ref{sec:proofmain} to conclude the proof of Theorem \ref{thm:SSBCondensate}. Since the perturbed grand potential is a concave function of $\lambda$, one can apply a Griffith (Hellmann--Feynman) argument to obtain bounds on the quasi-averages by differentiating the grand potential with respect to $\lambda$. This, however, only yields bounds for the expectation of $a_0 + a_0^*$ in the perturbed Gibbs state. To conclude the desired statement in \eqref{eq:SSBCondensate}, one has to combine these estimates with \eqref{eq:continuityCondensateFraction}.

 Finally, Appendix~\ref{app:effectiveChemicalPotential} is devoted to the analysis of equation \eqref{eq:GrantCanonicalEffectiveIddealGasChemPot} and the derivation of the basic properties of the effective chemical potential. 
 
\section{Expected number of particles in the unperturbed system}
\label{sec:expectednumberproof}
The goal of this section is to prove Proposition \ref{prop:expnumberparticles}. The proof will follow from an upper and a lower bound on the (unperturbed) grand potential  \eqref{eq:grantPotential}. 

\subsection{Upper bound for the unperturbed grand potential}

We start with an upper bound. To this end recall the exponential property of the Fock space $\mathscr{F}$ over $L^2(\Lambda)$, which states the unitary equivalence $\mathscr{F} \cong \mathscr{F}_0 \otimes \mathscr{F}_+$, where $\mathscr{F}_0$ denotes the Fock space over the one-dimensional Hilbert space $\text{span}(1)$ and $\mathscr{F}_+$ is the Fock space over $\mathds{1}(-\Delta \neq 0) L^2(\Lambda)$. We will use the Gibbs variational principle \eqref{eq:grantPotential} and our trial state will be of the form  $|z \rangle \langle z | \otimes G$, where $|z \rangle$ denotes the coherent state 
\begin{equation}
	| z \rangle = \exp( z a_0^* - \overline{z} a_0 ) | \Omega_0 \rangle, \quad \quad z \in \mathbb{C}
	\label{eq:coherentstate}
\end{equation}
with the vacuum vector $\Omega_0 \in \mathscr{F}_0$. The coherent state describes a BEC with an expected number of $|z|^2$ particles in the constant function $z/|z| \in L^2(\Lambda)$. $G$ is a state on the excitation Fock space $\mathscr{F}_+$ that describes the thermally excited particles in a non-interacting gas. 

\begin{proposition}[Upper bound for the unperturbed grand potential]\label{prop:grandpotupperbound}
    Assume that $\beta^{-1} \sim \eta^{2/3}$ and that $\mu$ and the interaction potential $v$ satisfy the assumptions of Theorem \ref{prop:expnumberparticles}. Let  $N_0(\beta,\widetilde{\mu})$ be given by \eqref{eq:crittemp} for the ideal gas with chemical potential $\widetilde{\mu}$ that is the solution of \eqref{eq:GrantCanonicalEffectiveIddealGasChemPot}.  Then 
    \begin{equation}\label{eq:grandpotupperbound}
        \Phi(\beta,\mu)\leq  \Phi_+^{\mathrm{id}} (\beta,\widetilde{\mu}) - \frac{(\mu-\widetilde{\mu})^2\eta}{2\hat{v}(0)}   + C\eta^{2/3}
    \end{equation}
    for $\eta$ large enough and some constant $C>0$.
\end{proposition}
\begin{proof}
Our trial state is given by
\begin{equation}
\Gamma^{\mathrm{trial}} = |   \sqrt{N_0(\beta,\widetilde{\mu})} \rangle \langle   \sqrt{N_0(\beta,\widetilde{\mu})} | \otimes G_+^{\mathrm{id}}(\beta,\widetilde{\mu})
\label{eq:trialstate}
\end{equation}
where
\begin{equation}
    G_+^{\mathrm{id}}(\beta,\widetilde{\mu}) =   \frac{\exp\left( -\beta (\de \Upsilon(-Q \Delta) - \widetilde{\mu} \mathcal{N}_+) \right)}{\tr_{\mathscr{F}_+} \exp\left( -\beta (\de \Upsilon(-Q \Delta) - \widetilde{\mu} \mathcal{N}_+) \right)}.
	\label{eq:boundChemicalPotential4}
\end{equation}
 The choice of $\widetilde{\mu}$ (recall \eqref{eq:GrantCanonicalEffectiveIddealGasChemPot}) implies  that
\begin{equation}
\label{eq:exptrial}    \tr[a_0^* a_0 \Gamma^{\mathrm{trial}}]=N_0(\beta,\widetilde{\mu}) \qquad \text{and} \qquad  \tr[\cN_+ \Gamma^{\mathrm{trial}}]=\frac{(\mu-\widetilde{\mu})\eta}{\hat{v}(0)}-N_0(\beta,\widetilde{\mu}).\end{equation}
We compute
\begin{align}
    \mathcal{H}_{\eta} =& \sum_{p \in \Lambda^*_+} (p^2-\widetilde{\mu}) a_p^* a_p +\widetilde{\mu}\cN_+ +\frac{\hat{v}(0)}{2\eta}a_0^* a_0^* a_0 a_0 + \frac{\hat{v}(0)}{\eta}\sum_{u \in \Lambda_+^*} a_u^* a_0^* a_u a_0 + \frac{\hat{v}(0)}{2\eta}\sum_{u,v \in \Lambda_+^*} a_u^* a_v^* a_u a_v \nonumber \\ &+ \frac{1}{2\eta} \sum_{p \in \Lambda^*_+} \hat{v}(p) \left\{ 2 a_p^* a_0^* a_p a_0 + a_0^* a_0^* a_p a_{-p} + a_p^* a_{-p}^* a_0 a_0 \right\} 
	+ \frac{1}{\eta} \sum_{p,k,p+k \in \Lambda_+^*} \hat{v}(p) \left\{ a^*_{k+p} a^*_{-p} a_k a_0 + h.c.   \right\}  \nonumber \\
	&+ \frac{1}{2\eta} \sum_{u,v,p,u+p,v-p \in \Lambda^*_+} \hat{v}(p) a^*_{u+p} a^*_{v-p} a_u a_v.
	\label{eq:decompH}
\end{align}
Since on the excited space our trial state is quasi-free and particle number conserving, it follows that
$$\tr[a_0^* a^*_0 a_p a_{-p} \Gamma^{\mathrm{trial}}]=0 \qquad \text{and} \qquad \tr[ a^*_{k+p} a^*_{-p} a_k a_0 \Gamma^{\mathrm{trial}}]=0$$
for any $p, k \in \Lambda_+^*$ with $p+k\neq 0$. On the other hand, using \eqref{eq:exptrial}, we see that
 \begin{equation*}
  \frac{\hat{v}(0)}{2\eta}\tr [a_0^* a_0^* a_0 a_0 \Gamma^{\mathrm{trial}}] + \frac{\hat{v}(0)}{\eta}\sum_{u \in \Lambda_+^*}\tr [a_0^* a_u^* a_0 a_u \Gamma^{\mathrm{trial}}]=\frac{\hat{v}(0)}{2\eta}N_0(\beta,\widetilde{\mu})^2+\frac{\hat{v}(0)}{\eta}N_0(\beta,\widetilde{\mu})\left(\frac{(\mu-\widetilde{\mu})\eta}{\hat{v}(0)}-N_0(\beta,\widetilde{\mu})\right). 
 \end{equation*}
In order to estimate the other terms, we introduce for $p,q \in \Lambda_+^*$ the notation $\gamma^{\mathrm{id}}(p,q) = \tr [ a_q^* a_p G_+^{\mathrm{id}}(\beta,\widetilde{\mu}) ]$. By translation invariance we have   $\gamma^{\mathrm{id}}(p,q)=\delta_{p,q}\gamma^{\mathrm{id}}(p)$, 
where 
$$\gamma^{\mathrm{id}}(p)=\frac{1}{\exp(\beta \left( p^2 -  \widetilde{\mu} ) \right) - 1}.$$
Using Wick's theorem again we get 
\begin{equation*}
 \sum_{u,v \in \Lambda_+^*} \tr[ a_u^* a_v^* a_u a_v \Gamma^{\mathrm{trial}}] =\sum_{u,v \in \Lambda_+^*} \gamma^{\mathrm{id}}(u)\gamma^{\mathrm{id}}(v)+\sum_{u \in \Lambda_+^*} \left(\gamma^{\mathrm{id}}(u)\right)^2\leq \left( \frac{(\mu-\widetilde{\mu})\eta}{\hat{v}(0)}-N_0(\beta,\widetilde{\mu})\right)^2+C\beta^{-2},\end{equation*}
where, in order to obtain the bound for the last term, we  used $\gamma^{\mathrm{id}}(p)\leq \beta^{-1}p^{-2}$ for $p\in \Lambda_+^*$.  This concludes the estimates for the last three terms in the first line of \eqref{eq:decompH}. 

We shall now estimate  the first term in the second line of \eqref{eq:decompH}. We see that 
\begin{equation}\label{eq:boundQQPP}
    \sum_{u \in \Lambda_+^*} \hat{v}(u) \tr[ a_u^* a_u a_0^* a_0 \Gamma^{\mathrm{trial}}]=N_0(\beta,\widetilde{\mu})\sum_{u \in \Lambda_+^*} \hat{v}(u)\gamma^{\mathrm{id}}(u)\leq C\beta^{-1} N_0(\beta,\widetilde{\mu}), 
\end{equation}
where we used the summability of $\hat{v}$ and $ \gamma^{\mathrm{id}}(p)\leq C\beta^{-1}$ for $p\in \Lambda_+^*$. Finally, estimating in a similar way the last term in \eqref{eq:decompH}, we get 
\begin{equation} \label{eq:boundQQQQ}
    \sum_{u,v,p,u+p,v-p \in \Lambda^*_+}\hat{v}(p) \tr[ a^*_{u+p} a^*_{v-p} a_u a_v \Gamma^{\mathrm{trial}}]= \sum_{v,p,v-p \in \Lambda^*_+}\hat{v}(p)\gamma^{\mathrm{id}}(v-p)\gamma^{\mathrm{id}}(v)\leq C\beta^{-1}\left( \frac{(\mu-\widetilde{\mu})\eta}{\hat{v}(0)}-N_0(\beta,\widetilde{\mu})\right).
    \end{equation}
Thus both, \eqref{eq:boundQQPP} and \eqref{eq:boundQQQQ}, are bounded by $C\beta^{-1}\eta$. Together with the fact that the corresponding terms in \eqref{eq:decompH} are coupled with $\eta^{-1}$, this implies that the contributions from these terms will be of order $\beta^{-1}=O(\eta^{2/3})$ and thus will be absorbed as an error term.     

Putting  these estimates together we obtain
\begin{equation}
  \begin{aligned}
      \tr[(\mathcal{H}_\eta-\mu \cN) \Gamma^{\mathrm{trial}}]\leq \tr[(d\Upsilon(Q(-\Delta-\widetilde{\mu})) G_+^{\mathrm{id}}(\beta,\widetilde{\mu})]-\frac{(\mu-\widetilde{\mu})^2\eta}{\hat{v}(0)}-\widetilde{\mu}N_0(\beta,\widetilde{\mu})+\frac{(\mu-\widetilde{\mu})^2\eta}{2\hat{v}(0)}+C\eta^{2/3}.
 \end{aligned}  \label{eq:upperboundidealtrialenergy}
\end{equation}
Here $\de \Upsilon(A)$ denotes the second quantization of the one-particle operator $A$. Using the fact that 
the entropy of $\Gamma^{\mathrm{trial}}$ satisfies $S(\Gamma^{\mathrm{trial}}) = S(G_+^{\mathrm{id}}(\beta,\widetilde{\mu}))$, we obtain the following upper bound
\begin{equation}
  \begin{aligned}
      \tr[(\mathcal{H}_\eta-\mu\cN) \Gamma^{\mathrm{trial}}]-\frac{1}{\beta}S(\Gamma^{\mathrm{trial}}) &\leq \Phi_+^{\mathrm{id}} (\beta,\widetilde{\mu}) - \frac{(\mu-\widetilde{\mu})^2\eta}{2\hat{v}(0)} -\widetilde{\mu}N_0(\beta,\widetilde{\mu})+ C \eta^{2/3},  
  \end{aligned}  \label{eq:upperboundidealtrial}
\end{equation}
where we used that 
$$ \tr[(d\Upsilon(Q(-\Delta-\widetilde{\mu})) G_+^{\mathrm{id}}(\beta,\widetilde{\mu})]-\frac{1}{\beta}S(G_+^{\mathrm{id}}(\beta,\widetilde{\mu}))=\Phi_+^{\mathrm{id}} (\beta,\widetilde{\mu}).$$
Since $\widetilde{\mu}N_0(\beta,\widetilde{\mu})\leq \beta^{-1}$ we obtain the final result.
\end{proof}

\begin{remark}
Upper bounds for a grand canonical version of the free energy in the Gross--Pitaevskii scaling (in the same temperature regime) have been obtained in \cite{DeuSeiYng-19,DeuSei-20,BocDeuSto-24,CapDeu-23}. 
\end{remark}

\subsection{Lower bound for the unperturbed grand potential and proof of Proposition \ref{prop:expnumberparticles}}
We start by recalling a well-known lemma, which gives a simple lower bound on the interaction term in the Hamiltonian.
\begin{lemma}\label{lem:OnsagersLemma}
	Let $v \in L^1(\Lambda)$ be a periodic function with summable Fourier coefficients $\hat{v} \geq 0$. Denote the second term in \eqref{eq:Hamiltonian} by $\mathcal{V}_\eta$. Then we have 
	\begin{equation}
		\mathcal{V}_\eta \geq \frac{\hat{v}(0) \mathcal{N}^2}{2\eta} - \frac{v(0) \mathcal{N}}{2\eta}.
		\label{eq:OnsagersInequality}
	\end{equation}
\end{lemma}
\begin{proof}
We compute
\begin{equation*}
    \begin{aligned}
    \sum_{p,u,v \in \Lambda^*} \hat{v}(p) a_{u+p}^* a_{v-p}^* a_u a_v & 
     =\sum_{u,v \in \Lambda^*} \hat{v}(0) a_{u}^* a_u a_{v}^* a_v -\sum_{p,u \in \Lambda^*} \hat{v}(p) a_{u+p}^*  a_{u+p}+\sum_{p \in \Lambda^*_+, \, u,v \in \Lambda^*} \hat{v}(p) a_{u+p}^* a_u a_{v-p}^* a_v \\
    &= \hat{v}(0) \mathcal{N}^2-v(0) \mathcal{N}+\sum_{p \in \Lambda^*_+} \hat{v}(p)B_p^* B_p,
\end{aligned}
\end{equation*}
where we introduced the notation $B_p=\sum_{u \in \Lambda^*}a_{u+p}^* a_u$. 	Since $\hat{v} \geq 0$ we can drop the last term for a lower bound and obtain the desired result.
\end{proof}

We are ready to state a lower bound for the unperturbed grand potential.

\begin{proposition}[Lower bound for the unperturbed grand potential]\label{prop:grandpotlowerbound}
    Assume that $\beta,\mu,v$ satisfy the assumptions of Theorem \ref{thm:SSBCondensate}. Let  $N_0(\beta,\widetilde{\mu})$ be given by \eqref{eq:crittemp} for the ideal gas with chemical potential $\widetilde{\mu}$ that is the solution of \eqref{eq:GrantCanonicalEffectiveIddealGasChemPot}.  Then for any bosonic state $\Gamma$ we have
    \begin{equation}\label{eq:grandpotowerbound}
        \mathcal{G}(\Gamma) \geq  \Phi_+^{\mathrm{id}} (\beta,\widetilde{\mu}) - \frac{(\mu-\widetilde{\mu})^2\eta}{2\hat{v}(0)}+\frac{\hat{v}(0)}{2\eta}\tr\left[\left(\cN-\frac{(\mu-\widetilde{\mu}+\frac{v(0)}{2\eta})\eta}{\hat{v}(0)}\right)^2 \Gamma\right]-C\eta^{2/3} \ln\eta 
    \end{equation}
    for $\eta$ large enough and some constant $C>0$.
\end{proposition}
\begin{proof}
   It follows from Lemma \ref{lem:OnsagersLemma} that
   \begin{equation}
   \begin{aligned}
       \tr[(\mathcal{H}_\eta-\mu\cN) \Gamma ]&-\frac{1}{\beta}S(\Gamma ) \geq \tr\left[\left(d\Upsilon(-\Delta-\mu)+\frac{\hat{v}(0)}{2\eta} \cN^2-\frac{v(0)}{2\eta}\cN\right)\Gamma\right]-\frac{1}{\beta}S(\Gamma ) \\
       & \geq \Phi^{\mathrm{id}} (\beta,\widetilde{\mu})+\tr\left[\left(\frac{\hat{v}(0)}{2\eta} \cN^2-\left(\mu-\widetilde{\mu}+\frac{v(0)}{2\eta}\right)\cN\right)\Gamma\right]\\
       &\geq  \Phi_+^{\mathrm{id}} (\beta,\widetilde{\mu})+\frac{1}{\beta} \ln\left( 1 - \exp(\beta \widetilde{\mu}) \right)+\frac{\hat{v}(0)}{2\eta}\tr\left[\left(\cN-\frac{(\mu-\widetilde{\mu}+\frac{v(0)}{2\eta})\eta}{\hat{v}(0)}\right)^2\Gamma\right]-\frac{(\mu-\widetilde{\mu}+\frac{v(0)}{2\eta})^2\eta}{2\hat{v}(0)}\\
       &\geq \Phi_+^{\mathrm{id}} (\beta,\widetilde{\mu}) - \frac{(\mu-\widetilde{\mu})^2\eta}{2\hat{v}(0)}+\frac{\hat{v}(0)}{2\eta}\tr\left[\left(\cN-\frac{(\mu-\widetilde{\mu}+\frac{v(0)}{2\eta})\eta}{\hat{v}(0)}\right)^2 \Gamma\right]-C\eta^{2/3} \ln\eta, 
   \end{aligned}
   \end{equation}
   where in the last step we used the fact that $\frac{1}{\beta} \ln\left( 1 - \exp(\beta \widetilde{\mu}) \right)\geq -C \eta^{2/3} \ln \eta$ (which follows easily from \eqref{eq:GrantCanonicalEffectiveIddealGasChemPot}). Taking $\Gamma=G_{\beta,\mu}$ ends the proof of the lemma.
\end{proof}
We are ready to prove Proposition \ref{prop:expnumberparticles}.

\begin{proof}[Proof of Proposition \ref{prop:expnumberparticles}.] Putting together\eqref{eq:grandpotupperbound} and \eqref{eq:grandpotowerbound} we obtain
$$\frac{\hat{v}(0)}{2\eta}\tr\left[\left(\cN-\frac{(\mu-\widetilde{\mu}+\frac{v(0)}{2\eta})\eta}{\hat{v}(0)}\right)^2 G_{\beta,\mu}\right]\leq C\eta^{2/3} \ln\eta $$
for some constant $C>0$. By the Cauchy-Schwarz inequality we know that
$$\tr\left[\left(\cN-\frac{(\mu-\widetilde{\mu}+\frac{v(0)}{2\eta})\eta}{\hat{v}(0)}\right)^2 G_{\beta,\mu}\right]\geq \left[\tr\left[\left(\cN-\frac{(\mu-\widetilde{\mu}+\frac{v(0)}{2\eta})\eta}{\hat{v}(0)}\right) G_{\beta,\mu}\right]\right]^2$$
and thus
$$  \left| \tr\left[\left(\cN-\frac{(\mu-\widetilde{\mu}+\frac{v(0)}{2\eta})\eta}{\hat{v}(0)}\right) G_{\beta,\mu}\right]\right|\leq C \eta^{5/6} \sqrt{\ln \eta}.$$
In particular
$$\left|N(\beta,\mu)-\frac{(\mu-\widetilde{\mu})\eta}{\hat{v}(0)}\right|\leq C\eta^{5/6} \sqrt{\ln \eta},$$
which concludes the proof.
\end{proof}

\begin{remark}
Note that the condition for $\widetilde{\mu}$ \eqref{eq:GrantCanonicalEffectiveIddealGasChemPot} appears naturally from the proof of Proposition  \ref{prop:grandpotlowerbound}. Indeed, repeating the same estimates but with a general unknown chemical potential $\widehat{\mu}$ leads to the lower bound
$$\tr[(\mathcal{H}_\eta-\mu\cN) \Gamma ]-\frac{1}{\beta}S(\Gamma ) \geq \Phi_+^{\mathrm{id}} (\beta,\widehat{\mu}) - \frac{(\mu-\widehat{\mu})^2\eta}{2\hat{v}(0)} -C\eta^{2/3} \ln\eta. $$
The condition \eqref{eq:GrantCanonicalEffectiveIddealGasChemPot} can then be obtained by maximizing the first two terms on the right hand side of the inequality over $\widehat{\mu}$. 
\end{remark}
\section{Bounds for the perturbed grand potential}
\label{sec:boundsgrandPotential}

The goal of this section is to provide upper and lower bounds for a family of perturbed grand potentials. More precisely, we will consider
\begin{equation}
    \mathcal{H}^{\lambda,\delta}_{ \eta} = \mathcal{H}^{\lambda}_{ \eta}+\delta a_0^* a_0 = \mathcal{H}_{ \eta} + \delta a_0^* a_0+  \lambda { N(\beta,\mu)}^{1/2} (a_0 + a_0^*)
    \label{eq:HamiltonianWithDoublePerturb}
\end{equation}
for two parameters $\delta,\lambda\in \mathbb{R}$. 
We introduce the corresponding Gibbs state
\begin{equation}
    G^{\lambda,\delta}_{\beta,{ \mu}} = \frac{\exp\left( -\beta (\mathcal{H}^{\lambda,\delta}_{ \eta} - \mu  \mathcal{N}) \right) }{\tr \exp\left( -\beta (\mathcal{H}^{\lambda,\delta}_{ \eta} - \mu \mathcal{N}) \right) }
    \label{eq:doubleperturbedGibbsstate}
\end{equation}
and grand potential
\begin{equation}
    \Phi^{\lambda,\delta}(\beta,\mu) =  -\frac{1}{\beta} \ln\left( \tr \exp(-\beta(\mathcal{H}^{\lambda,\delta}_{ \eta} - \mu \mathcal{N})) \right).
    \label{eq:perturbedgrantPotential}
\end{equation}

The main result of this section is the following estimate. 

\begin{proposition}[Bound the perturbed grand potential] \label{prop:grandpotperturbedupperlowerbound}
    Assume that $\beta,\mu,v$ satisfy the assumptions of Theorem \ref{thm:SSBCondensate}. Let  $N_0(\beta,\widetilde{\mu})$ be given by \eqref{eq:crittemp} for the ideal gas with chemical potential $\widetilde{\mu}$ that is the solution of \eqref{eq:GrantCanonicalEffectiveIddealGasChemPot}.  Then, for all $\delta,\lambda\in \mathbb{R}$ and $\eta$ large enough, we have
\begin{align}\label{eq:grandpotperturbedupperlowerbound}
       &\left|  \Phi^{\lambda,\delta}(\beta,\mu) - \left(  \Phi_+^{\mathrm{id}} (\beta,\widetilde{\mu}) - \frac{(\mu-\widetilde{\mu})^2\eta}{2\hat{v}(0)}  +\delta N_0(\beta,\widetilde{\mu})   - 2 |\lambda|\sqrt{N(\beta,\mu) N_0(\beta,\widetilde{\mu})} \right)\right| \nn\\
       &\quad \le C \eta (\delta^2 + |\delta| |\lambda|^{1/3} + |\lambda| ^{4/3} + \eta^{-1/6} \ln \eta).
           \end{align}
\end{proposition}

We will deduce \eqref{eq:grandpotperturbedupperlowerbound} from the two separate estimates, i.e. the upper and lower bounds. While the upper bound follows easily from the proof of Proposition \ref{prop:grandpotupperbound}, the lower bound requires a deeper investigation. The details are discussed in Propositions \ref{prop:grandpotperturbedupperbound} and \ref{prop:lowerboundgrandpotperturbed} below.

\subsection{Upper bound for the perturbed grand potential}

In this subsection we prove the following upper bound.

\begin{proposition}[Upper bound for the perturbed grand potential] \label{prop:grandpotperturbedupperbound}
    Assume that $\beta,\mu,v$ satisfy the assumptions of Theorem \ref{thm:SSBCondensate}. Let  $N_0(\beta,\widetilde{\mu})$ be given by \eqref{eq:crittemp} for the ideal gas with chemical potential $\widetilde{\mu}$ that is the solution of \eqref{eq:GrantCanonicalEffectiveIddealGasChemPot}.  Then, for all $\delta,\lambda\in \mathbb{R}$ and $\eta$ large enough, we have
    \begin{equation}\label{eq:grandpotperturbedupperbound}
        \Phi^{\lambda,\delta}(\beta,\mu)\leq  \Phi_+^{\mathrm{id}} (\beta,\widetilde{\mu}) - \frac{(\mu-\widetilde{\mu})^2\eta}{2\hat{v}(0)}  +\delta N_0(\beta,\widetilde{\mu})   - 2 |\lambda|\sqrt{N(\beta,\mu) N_0(\beta,\widetilde{\mu})} + C\eta^{2/3}   . 
    \end{equation}
\end{proposition}
\begin{proof}
The proof is similar to the proof of Proposition \ref{prop:grandpotupperbound}. However, due to the symmetry breaking term in the Hamiltonian, we need to modify our trial state so that it takes care of the phase of the perturbation.  More  precisely, our trial state is given by
\begin{equation}
\Gamma^{\lambda} = | -(\lambda/|\lambda|) \sqrt{N_0(\beta,\widetilde{\mu})} \rangle \langle -(\lambda/|\lambda|) \sqrt{N_0(\beta,\widetilde{\mu})} | \otimes G_+^{\mathrm{id}}(\beta,\widetilde{\mu}).
\label{eq:SSBTrialState}
\end{equation}
The identities \eqref{eq:exptrial} remain the same. In particular
$$\delta \tr[a_0^* a_0  \Gamma^{\lambda}]=\delta N_0(\beta,\widetilde{\mu}).$$
However, now  
$$\tr[a_0^*  \Gamma^{\lambda}]=\frac{-\lambda}{|\lambda|}\sqrt{N_0(\beta,\widetilde{\mu})}$$
which implies
$$\lambda \tr[(a_0^*+a_0)  \Gamma^{\lambda}]=- 2 |\lambda|\sqrt{N_0(\beta,\widetilde{\mu})}.$$
When multiplied by $\sqrt{N(\beta,\mu)}$ this gives the $\lambda$-dependent term on the right hand side of \eqref{eq:grandpotperturbedupperbound}. All the other computations remain the same as in the proof of Proposition \ref{prop:grandpotupperbound}. (Note, that the only term which could be affected by the phase is the last  term in the second line of \eqref{eq:decompH} but its contribution will again vanish as $\Gamma^{\lambda}$, when restricted to $\mathscr{F}_+$, is a quasi-free state that commutes with $\mathcal{N}_+$.). This ends the proof. 
\end{proof}


\subsection{C-number substitution and relative entropy} \label{ssec:cnumber}
In order to derive a lower bound for the grand potential, we will use the $c$-number substitution in the spirit of \cite{LieSeiYng-05}. Let us now briefly recall the main facts about this approach. We start with the resolution of identity
\begin{equation}
	\int_{\mathbb{C}} | z \rangle \langle z | \de z = \mathds{1}_{\mathscr{F}_0} 
	\label{eq:resolutionofidentity}
\end{equation}
on the Fock space $\mathscr{F}_0$ over the $p=0$ mode. Here
 $|z\rangle$ defined in \eqref{eq:coherentstate}  
denotes the coherent state indexed by the complex number $z = x+iy$ with $x,y \in \mathbb{R}$ and $ \de z = \pi^{-1} \de x \de y$ denotes the appropriately normalized Lebesgue measure on the complex plane. Given any state $\Gamma \in \mathcal{S}$ with $\mathcal{S}$ in \eqref{eq:states}, we define the operator $\widetilde{\Gamma}_z$ acting on the excitation Fock space $\mathscr{F}_+$ (defined above \eqref{eq:coherentstate}) by
\begin{equation}
	\widetilde{\Gamma}_z = \tr_{\mathscr{F}_0}[ | z \rangle \langle z | \Gamma ] = \langle z, \Gamma z \rangle
    \label{eq:cnumber1}
\end{equation} 
and we denote 
\begin{equation}
	\zeta_{\Gamma}(z) = \tr_{\mathscr{F}_+}[ \widetilde{ \Gamma }_z ].
    \label{eq:cnumber2}
\end{equation}
  Since $\Gamma$ is a state, $\zeta_{\Gamma}$ defines a probability measure on $\mathbb{C}$. By $S(\zeta_{\Gamma})$ we denote the entropy of the classical probability distribution $\zeta_{\Gamma}$, that is,
\begin{equation}
	S(\zeta_\Gamma) = - \int_{\mathbb{C}} \zeta_{\Gamma}(z) \ln\left( \zeta_{\Gamma}(z) \right) \de z.
    \label{eq:cnumber3}
\end{equation} 
We also define the state 
\begin{equation}
	\Gamma_z = \frac{\widetilde{\Gamma}_z}{\tr_{\mathscr{F}_+}[ \widetilde{\Gamma}_z ]}
 \label{eq:cnumber4}
\end{equation}
on $\mathscr{F}_+$. The following Lemma, whose proof can be found in \cite[Lemma~3.2]{DeuSei-20}, provides us with an upper bound for the entropy of $\Gamma$ in terms of the ones of $\Gamma_z$ and $\zeta_{\Gamma}$.
\begin{lemma}
	\label{lem:entropyinequality}
	Let $\Gamma$ be a state on $\mathscr{F}$. The entropy of $\Gamma$ is bounded in the following way:
	\begin{equation}
		S(\Gamma) \leq \int_{\mathbb{C}} S(\Gamma_z) \zeta_{\Gamma}(z) \de z + S(\zeta_{\Gamma}).
    \label{eq:cnumber5}
	\end{equation}
\end{lemma}
The above Lemma allows us to replace the entropy of $\Gamma$ in the grand potential functional by the ones of $\Gamma_z$ and $\zeta_{\Gamma}$ for a lower bound. In order to express also the energy in terms of the ones of $\Gamma_z$ and $\zeta_{\Gamma}$, we introduce the upper  symbol $\mathcal{H}^{\mathrm{s}}$ related to a general Hamiltonian $\mathcal{H}$. 
It is defined by the relation
\begin{equation}
	\mathcal{H} = \int_{\mathbb{C}} \mathcal{H}^{\mathrm{s}}(z) | z \rangle \langle z | \de z. \label{eq:cnumber6}
\end{equation}
For example, the upper symbols $(a_0)^s(z)$ and $(a^*_0)^s(z)$ of $a_0$ and $a_0^*$ are simply $z$ and $\bar{z}$, respectively. For further reference we notice that the upper symbols $\cN^{s}(z)$ of $\cN$ and $(\cN^2)^{s}(z)$ of $\cN^2$ are given by
\begin{equation}\label{eq:uppersymbolsnumber}
  \cN^{s}(z)=|z|^2-1+\cN_+,\qquad (\cN^2)^{s}(z)=(|z|^2+\cN_+)^2-3(|z|^2+\cN_+)+\cN_++1.
\end{equation}
 More information on the upper symbol can be found, e.g., in \cite{LieSeiYng-05}. Using \eqref{eq:cnumber6} we can write the expectation of the energy with respect to a state $\Gamma \in \mathcal{S}$ as
\begin{equation}
	\tr[\mathcal{H} \Gamma] = \int_{\mathbb{C}} \tr[ \mathcal{H}^{\mathrm{s}}(z) | z \rangle \langle z | \Gamma ] \de z = \int_{\mathbb{C}} \tr_{\mathscr{F}_+}[ \mathcal{H}^{\mathrm{s}}(z) \Gamma_z ] \zeta_{\Gamma}(z) \de z.
	\label{eq:cnumber7}
\end{equation}
This will be used in the next subsection. 

Another quantity that will be used in the proof below is the relative entropy. The relative entropy of a state $\Gamma$ with respect to another state $\Gamma'$ given by 
$$S(\Gamma, \Gamma')=\tr[\Gamma(\ln(\Gamma)-\ln(\Gamma'))].$$
It quantifies the difference between the grand potential functional evaluated at a given state $\Gamma$ and the grand potential (corresponding to a Gibbs state), i.e.
\begin{equation} \label{eq:relativeentropygrandpotential}
    S(\Gamma, \Gamma_0)=\beta (\mathcal{G}(\Gamma)-\mathcal{G}(\Gamma_0))
\end{equation}
where $\Gamma_0$ is a Gibbs state (cf. \eqref{eq:grandPotentialFunctional}). In the case when the underlying Hamiltonian is non-interacting, i.e. in the definition of \eqref{eq:grandPotentialFunctional} one has $\mathcal{H}=d\Upsilon(h)$ for a one-body operator $h$, then the relative entropy can be bounded from below by the so-called bosonic relative entropy  as done in \cite{DeuSeiYng-19,DeuSei-20}. Indeed, in this case we have
\begin{equation*} 
\beta (\mathcal{G}(\Gamma)-\mathcal{G}(\Gamma_0))=\beta\tr[h\gamma_{\Gamma}]-S(\Gamma)-\beta\tr[h\gamma_{\Gamma_0}]+S(\Gamma_0)
\end{equation*}
where by $\gamma_G$ we denote the one-body reduced density matrix of a state $G$ (cf. \eqref{eq:1pdm}). Since the Gibbs state of a non-interacting Hamiltonian is quasi-free, we have (see, e.g., \cite[Appendix A]{NapReuSol-18a}) that 
$$S(\Gamma_0)=-\tr \left[ \gamma_{\Gamma_0} \ln(\gamma_{\Gamma_0}) - (1+\gamma_{\Gamma_0}) \ln (1+\gamma_{\Gamma_0})\right]=: -\tr \left[ \sigma\left( \gamma_{\Gamma_0} \right) \right],$$
where we introduced
$$\sigma(x) = x \ln(x) - (1+x) \ln (1+x).$$
Now, using the fact (see \cite[2.5.14.5]{Thirring_4}) that
$$-\tr \left[ \sigma\left( \gamma_{\Gamma} \right) \right]\geq S(\Gamma),$$
we arrive at 
\begin{equation} \label{eq:relativeentropylowerbound}
S(\Gamma, \Gamma_0)\geq s(\gamma_{\Gamma},\gamma_{\Gamma_0}),
\end{equation}
where for two nonnegative operators $a, b$ with finite trace, the bosonic relative entropy $s(a,b)$ is defined by
\begin{equation}
s(a,b) = \sum_{i,j} \left| \langle \psi_i, \varphi_j \rangle \right|^2 \left( \sigma(\gamma_i) - \sigma(\eta_j) - \sigma'(\eta_j)(\gamma_i-\eta_j)  \right).
\label{eq:asymptotics1pdm7}
\end{equation}
Here $\{ \lambda_i, \psi_i \}$ and $\{ \eta_j, \varphi_j \}$ denote the eigenvalues and eigenfunctions of $a$ and $b$, respectively. 

In order to be able to use \eqref{eq:relativeentropylowerbound} one needs a lower bound on the bosonic relative entropy. This is provided by the following lemma whose proof can be found in  \cite[Lemma~4.1]{DeuSei-20} (see also \cite[Lemma~4.1]{DeuSeiYng-19}).
 \begin{lemma}\label{lem:lowerboundbosonicrelativentropy}
Assume $a$ and $b$ are  two nonnegative
trace-class operators and let $s(a,b)$ be given in \eqref{eq:asymptotics1pdm7}. There exists a constant $c_1\geq \frac{2}{27}$ such that 
$$ s(a, b) \geq  c_1 \frac{\|a-b\|_1^2}{\|1+b\| \tr[a+b]}.$$
 \end{lemma}

\subsection{Lower bound for the perturbed grand potential} \label{ssec:lowerboundgrandpot}

Using the tools introduced in the previous section, we will prove the following lower bound on the perturbed grand potential. 

\begin{proposition}[Lower bound for the perturbed grand potential]\label{prop:lowerboundgrandpotperturbed}
    Assume that $\beta,\mu$ and the interaction potential $v$ satisfy the assumptions of Theorem \ref{thm:SSBCondensate}. Let  $N_0(\beta,\widetilde{\mu})$ be given by \eqref{eq:crittemp} for the ideal gas with chemical potential $\widetilde{\mu}$ that is the solution of \eqref{eq:GrantCanonicalEffectiveIddealGasChemPot}.  Then for all $\delta, \lambda\in \mathbb{R}$ and $\eta$ large enough, we have 
  \begin{equation}\label{eq:grandpotlowerboundperturbed}
  \begin{aligned}
       \Phi^{\lambda,\delta}(\beta,\mu)&\geq  \Phi_+^{\mathrm{id}} (\beta,\widetilde{\mu}) -   \frac{(\mu-\widetilde \mu)^2\eta}{2\hat v(0)}  +  \delta N_0(\beta,\widetilde \mu)  - 2|\lambda| \sqrt{N(\beta,\mu) N_0(\beta,\widetilde \mu)} \nn\\
&\quad - C \eta (\delta^2 + |\delta| |\lambda|^{1/3} + |\lambda| ^{4/3} + \eta^{-1/3} \ln \eta ).  
\end{aligned}
       \end{equation}
\end{proposition}
\begin{proof}
Using Lemma \ref{lem:OnsagersLemma}, for any state $\Gamma\in \mathcal{S}$ we have
\begin{equation*}
\tr[ (\mathcal{H}_\eta^{\lambda,\delta} - \mu \mathcal{N}) \Gamma ]  \geq 
\tr\left[ \left(\sum_{p\in \Lambda}(p^2-\mu)a_p^* a_p + \delta a_0^* a_0 + \lambda \sqrt{N(\beta,\mu)} (a_0 + a_0^*) +\frac{\hat{v}(0)}{2\eta}\mathcal{N}^2-\frac{v(0)}{2\eta}\mathcal{N}\right) \Gamma \right]. 
\end{equation*}
Rewriting the right-hand side in terms of the upper symbol (cf. \eqref{eq:cnumber7}) and using \eqref{eq:uppersymbolsnumber} we obtain
\begin{equation*}
\tr[ (\mathcal{H}_\eta^{\lambda,\delta} - \mu \mathcal{N}) \Gamma ]  \geq \int_{\mathbb{C}}  \tr_{\mathscr{F}_+} [ H^{\lambda}(z) \Gamma_z] \zeta_{\Gamma}(z) \de z -C \eta^{2/3}
\end{equation*}
with
\begin{equation} \label{eq:grandpotsymbol}
H^{\lambda,\delta}(z) = \de \Upsilon(-Q \Delta) - \left( \mu + \frac{v(0)}{2\eta} + \frac{3 \hat{v}(0)}{2\eta} \right) ( |z|^2 + \mathcal{N}_+ ) + \frac{\hat{v}(0)}{2\eta} (|z|^2 + \mathcal{N}_+ )^2 +\delta (|z|^2-1)+ \lambda \sqrt{N(\beta,\mu)}( z + \overline{z})
\end{equation}
and $Q = \mathds{1}(-\Delta \neq 0)$. Here we used the assumption that $\mu\geq -C\eta^{2/3}$ for some $C>0$. Together with Lemma \ref{lem:entropyinequality} this implies
\begin{equation}\label{eq:firstlowerboundgrandpot} 
\tr[ (\mathcal{H}_\eta^{\lambda,\delta} - \mu \mathcal{N}) \Gamma ] - \frac{1}{\beta} S(\Gamma) \geq \int_{\mathbb{C}} \left\{ \tr_{\mathscr{F}_+} [ H^{\lambda,\delta}(z) \Gamma_z] - \beta^{-1} S(\Gamma_z) \right\} \zeta_{\Gamma}(z) \de z - \beta^{-1} S(\zeta_{\Gamma})-C\eta^{2/3}.
\end{equation}
 Using \eqref{eq:relativeentropygrandpotential} we obtain
\begin{align}
\int_{\mathbb{C}} \left\{ \tr_{\mathscr{F}_+}[ \de \Upsilon(-Q \Delta) \Gamma_z ] - \beta^{-1} S(\Gamma_z) \right\} \zeta_{\Gamma}(z) \de z =& \Phi_+^{\mathrm{id}} (\beta,\widetilde{\mu} )+ \widetilde{\mu} \int_{\mathbb{C}} \tr_{\mathscr{F}_+} [ \mathcal{N}_+ \Gamma_z ] \ \zeta_{\Gamma}(z) \de z \nonumber \\
&+ \frac{1}{\beta} \int_{\mathbb{C}} S(\Gamma_z, G_+^{\mathrm{id}}(\beta,\widetilde{\mu})) \ \zeta_{\Gamma}(z) \de z,
\label{eq:SSBCondensate3}
\end{align}
where $ G_+^{\mathrm{id}}(\beta,\widetilde{\mu})$ denotes the grand canonical Gibbs state of the ideal gas with chemical potential $\widetilde{\mu}$ restricted to the excited Fock space with a partial trace (cf. \eqref{eq:boundChemicalPotential4}).

Let $\gamma_z$ denote the one-particle density matrix of $\Gamma_z$ and $\gamma^{\mathrm{id}}_+$ the one of  $ G_+^{\mathrm{id}}(\beta,\widetilde{\mu})$ (which is a quasi-free state). Applying \eqref{eq:relativeentropylowerbound} we obtain
 \begin{equation}
 	S\left(\Gamma_z,G^{\mathrm{id}}_+(\beta,\widetilde{\mu})\right) \geq s\left( \gamma_z, \gamma^{\mathrm{id}}_+ \right),
 	\label{eq:asymptotics1pdm8b}
 \end{equation}
which, using Lemma \ref{lem:lowerboundbosonicrelativentropy}, leads to
\begin{equation}
s(\gamma_z,\gamma^{\mathrm{id}}_+) \geq c_1 \frac{ \Vert \gamma_z - \gamma^{\mathrm{id}}_+ \Vert_1^2}{\Vert 1 +\gamma^{\mathrm{id}}_+ \Vert \,\Tr_{\mathscr{F}_+}[\gamma_z +\gamma^{\mathrm{id}}_+]}. 
\label{eq:SSBCondensate3b}
\end{equation}
Here $c_1$ is the constant from Lemma \ref{lem:lowerboundbosonicrelativentropy}, i.e. $c_1\geq \frac{2}{27}$.  The operator norm on the right hand side satisfies $\Vert 1 + \gamma^{\mathrm{id}}_+ \Vert  \lesssim \beta^{-1}$. Let us define 
$$\Tr_{\mathscr{F}_+} \gamma_z =: N_+(z),\qquad  N_+(\beta,\widetilde \mu) := \Tr_{\mathscr{F}_+}[ \gamma^{\mathrm{id}}_+]=\Tr_{\mathscr{F}_+}[\cN_+ G_+^{\mathrm{id}}(\beta,\widetilde{\mu})].$$
Then an application of the triangle inequality of the trace-norm shows
\begin{equation}
\frac{\Vert \gamma_z - \gamma^{\mathrm{id}}_+ \Vert_1^2}{\Tr_{\mathscr{F}_+}[\gamma_z + \gamma^{\mathrm{id}}_+]} \geq \frac{| N_+(z) - N_+ |^2}{N_+(z) + N_+} \geq \left( \sqrt{N_+(z)} - \sqrt{N_+} \right)^2,
\label{eq:SSBCondensate3c}
\end{equation}
and hence
\begin{align}
\frac{1}{\beta} \int_{\mathbb{C}} S(\Gamma_z,G^{\mathrm{id}}_+(\beta,\widetilde{\mu})) \ \zeta_{\Gamma}(z) \de z \gtrsim \int_{\mathbb{C}} \left( \sqrt{N_+(z)} - \sqrt{N_+} \right)^2 \zeta_{\Gamma}(z) \de z. 
\label{eq:SSBCondensate3d}
\end{align}
An application of the Gibbs variational principle (classical case) shows
\begin{equation}
-\frac{1}{\beta} S(\zeta_{\Gamma}) \geq - \frac{1}{\beta} \ln \left( \int_{\mathbb{C}} \exp(\beta \widetilde{\mu} |z|^2 ) \de z \right) + \widetilde{\mu} \int_{\mathrm{C}} |z|^2 \ \zeta_{\Gamma}(z) \de z\geq -C\beta^{-1}\ln \eta + \widetilde{\mu} \int_{\mathrm{C}} |z|^2 \ \zeta_{\Gamma}(z) \de z
\label{eq:SSBCondensate4}
\end{equation}
for some $C>0$. In the last step we used \eqref{eq:muMinusMu0Bound1} and our assumption on $\mu$.  Next, we take a closer look at the interaction term (the one originating from $\cN^2$) in \eqref{eq:grandpotsymbol}. By the Cauchy-Schwarz inequality we have
\begin{equation}
\int_{\mathbb{C}} \tr_{\mathscr{F}_+}[ (|z|^2 + \mathcal{N}_+ )^2 \Gamma_z] \ \zeta_{\Gamma}(z) \de z \geq \int_{\mathbb{C}} [ |z|^2 + N_+(z) ]^2 \ \zeta_{\Gamma}(z) \de z.
\label{eq:SSBCondensate5}
\end{equation}
Thus, in combination, \eqref{eq:firstlowerboundgrandpot}, \eqref{eq:SSBCondensate3}, \eqref{eq:SSBCondensate3d}, \eqref{eq:SSBCondensate4}, and \eqref{eq:SSBCondensate5} show that for some $c_1 \geq \frac{2}{27}$ we have
\begin{align}
\tr[ (\mathcal{H}_\eta^{\lambda,\delta} - \mu \mathcal{N}) \Gamma ] - \frac{1}{\beta} S(\Gamma) \geq& \Phi_+^{\text{id}}(\beta,\widetilde{\mu}) - C \beta^{-1} \ln(\eta) + \int_{\mathbb{C}} \bigg\{ c_1 \left( \sqrt{N_+(z)} - \sqrt{N_+(\beta,\widetilde \mu) } \right)^2 + \frac{\hat{v}(0)}{2\eta} [|z|^2 + N_+(z)]^2 \nn \\
&  - (\mu-\widetilde{\mu} + \eta^{-1} d_v ) (|z|^2 + N_+(z)) +\delta |z|^2+ \lambda \sqrt{N(\beta,\mu)}( z + \overline{z}) \bigg\} \zeta_{\Gamma}(z) \de z, \label{eq:SSBCondensate6} 
\end{align}
where $d_v=\frac12(v(0)+3\hat{v}(0))$. To obtain a lower bound for the right-hand side of \eqref{eq:SSBCondensate6}, we will minimize the expression inside the curly brackets  on the right hand side of \eqref{eq:SSBCondensate6} treated as a function of two independent variables $\sqrt{N_+(z)}$ and $|z|$ (to be precise the expression in the curly brackets depends also on $w:=\cos(\arg z)$, but this dependence is trivial as the phase will always want to be $-\sgn(\lambda)$).

Before that, let us rescale the relevant terms. It follows from \eqref{eq:GrantCanonicalEffectiveIddealGasChemPot} and Proposition \ref{prop:expnumberparticles} that the quantities  
\begin{equation}\label{eq:rescaledvariablesminimization}
    n_0:=\eta^{-1} N_0(\beta,\widetilde \mu),\quad n_+:=\eta^{-1}N_+(\beta,\widetilde \mu),\qquad  n(\beta,\mu):=\eta^{-1} N(\beta,\mu)
\end{equation}
are uniformly bounded as $\eta \to \infty$. We also have 
\begin{equation}\label{eq:rescaledvariablesminimization-b}
n_0+n_+=\frac{\mu - \widetilde{\mu}}{\hat{v}(0)}=n(\beta,\mu) + O(\eta^{-1/6}\ln \eta).
\end{equation}
Introducing the variables 
\begin{equation}\label{eq:rescaledvariablesminimizationxy}
   x:= x(z)=\sqrt{N_+(z)\eta^{-1}}, \qquad y:=y(z)=|z|\eta^{-1/2},
\end{equation}
we deduce from \eqref{eq:SSBCondensate6}  that 
\begin{align}\label{eq:SSBCondensate6-b}
\tr[ (\mathcal{H}_\eta^{\lambda,\delta} - \mu \mathcal{N}) \Gamma ] - \frac{1}{\beta} S(\Gamma) \geq \Phi_+^{\text{id}}(\beta,\widetilde{\mu}) - C \beta^{-1} \ln(\eta) + \eta \int_{\mathbb{C}}   F^{\lambda,\delta}(x(z),y(z)) \zeta_{\Gamma}(z) \de z,  
\end{align}
where 
\begin{equation}\label{eq:Flambdelxy}
  F^{\lambda,\delta}(x,y)=  c_1 \left( x - \sqrt{n_+} \right)^2 + \frac{\hat{v}(0)}{2}  (x^2+y^2)^2 - (\mu-\widetilde{\mu} + \eta^{-1} d_v)   (x^2+y^2) +\delta y^2-  2|\lambda| \sqrt{n(\beta,\mu)} y.
\end{equation}

We will derive a lower bound for $  F^{\lambda,\delta}(x,y)$ with $x,y\ge 0$. We consider three cases: 
$$\left\{x^2 \ge \frac{2(\mu-\widetilde \mu +\eta^{-1}d_v)}{\hat v(0)} \right\},\quad \Big\{y\le \sqrt{n_0}+|\lambda|^{1/3}\Big\},\quad \left\{x^2 \le \frac{2(\mu-\widetilde \mu +\eta^{-1}d_v)}{\hat v(0)} \right\} \cap \Big\{y\ge \sqrt{n_0}+|\lambda|^{1/3}\Big\}.$$  

\textbf{Case 1:} If $x^2 \ge 2(\mu-\widetilde \mu +\eta^{-1}d_v)/\hat v(0)$, then we use $(x^2+y^2)^2 \ge x^2 (x^2+y^2)+y^4$ to estimate
$$
 \frac{\hat{v}(0)}{2} x^2 (x^2+y^2) \ge (\mu-\widetilde{\mu} + \eta^{-1} d_v) (x^2+y^2).  
$$
Therefore, by dropping $c_1 ( x - \sqrt{n_+})^2 \ge 0$, we get
\begin{align}\label{eq:Fxy-c0-final-pre}
  F^{\lambda,\delta}(x,y) \ge \frac{\hat{v}(0)}{2} y^4 + \delta y^2 - 2 |\lambda| \sqrt{n(\beta,\mu)} y \ge - C (\delta^2 + |\lambda|^{4/3}). 
\end{align}
Here we used the Cauchy--Schwarz and H\"older  inequalities. Moreover,   
\begin{align}\label{eq:Fxy-c0-final-pre2}
\frac{(\mu-\widetilde \mu)^2}{2\hat v(0)} - \delta n_0 \ge - \frac{\hat v(0)n_0^2}{2(\mu-\widetilde \mu)^2} \delta^2 \ge - C\delta^2
\end{align}
by Cauchy--Schwarz and \eqref{eq:muMinusMu0Bound1}. In combination, \eqref{eq:Fxy-c0-final-pre} and \eqref{eq:Fxy-c0-final-pre2} show
\begin{align}\label{eq:Fxy-c0-final}
  F^{\lambda,\delta}(x,y) \ge  -\frac{(\mu-\widetilde \mu)^2}{2\hat v(0)} +\delta n_0  - C (\delta^2 + |\lambda|^{4/3}). 
\end{align}

\textbf{Case 2:} If $y\le \sqrt{n_0}+|\lambda|^{1/3}$ we have 
\begin{align}\label{eq:Fxy-c1-1}
\delta y^2- 2|\lambda| \sqrt{n(\beta,\mu)} y &= \delta n_0 - 2|\lambda| \sqrt{n(\beta,\mu)} \sqrt{n_0} + \delta ( y^2 -n_0) - 2|\lambda| \sqrt{n(\beta,\mu)} (y-\sqrt{n_0}) \nn\\
&\ge  \delta n_0 -  2|\lambda| \sqrt{n(\beta,\mu) n_0}  - |\delta|  \left(\sqrt{n_0}+\lambda^{1/3})^2-n_0 \right) - 2 |\lambda|^{4/3}   \sqrt{n(\beta,\mu)}\nn\\
&\ge \delta n_0 -  2|\lambda| \sqrt{n(\beta,\mu) n_0}  - C (|\delta| |\lambda|^{1/3}+|\delta| |\lambda|^{2/3}+ |\lambda|^{4/3}). 
\end{align}
Moreover, by completing the square, 
\begin{align}\label{eq:Fxy-c1-2}
 \frac{\hat{v}(0)}{2}  (x^2+y^2)^2 - (\mu-\widetilde{\mu} + \eta^{-1} d_v)   (x^2+y^2) & = -\frac{(\mu-\widetilde \mu + \eta^{-1}d_v)^2}{2\hat v(0)} +  \frac{\hat{v}(0)}{2} \left(x^2+y^2 - \frac{\mu-\widetilde{\mu}+\eta^{-1} d_v}{\hat{v}(0)}  \right)^2 \nn \\
 &\ge -\frac{(\mu-\widetilde \mu + \eta^{-1}d_v)^2}{2\hat v(0)} \ge -\frac{(\mu-\widetilde \mu)^2}{2\hat v(0)}  - C \eta^{-1}. 
\end{align}
To obtain the last bound, we used \eqref{eq:muMinusMu0Bound1}. Putting \eqref{eq:Fxy-c1-1} and \eqref{eq:Fxy-c1-2} together and dropping $c_1 ( x - \sqrt{n_+})^2 \ge 0$ again, we conclude that 
\begin{align}\label{eq:Fxy-c1-final}
  F^{\lambda,\delta}(x,y) \ge  -\frac{(\mu-\widetilde \mu)^2}{2\hat v(0)}  +  \delta n_0 -  2|\lambda| \sqrt{n(\beta,\mu) n_0}   - C (|\delta| |\lambda|^{1/3}+|\delta| |\lambda|^{2/3}+ |\lambda|^{4/3}+ \eta^{-1}). 
\end{align}

\textbf{Case 3:} If $x^2 \le 2(\mu-\widetilde \mu +\eta^{-1}d_v)/\hat v(0)$ and $y\ge \sqrt{n_0}+|\lambda|^{1/3}$, then we have
\begin{align}\label{eq:Fxy-c3-1}
c_1 (x-\sqrt{n_+})^2= c_1 \frac{(x^2-n_+)^2}{(x+\sqrt{n_+})^2} \ge c_2 (x^2-n_+)^2
\end{align}
for a constant $c_2 > 0$, and similarly
\begin{align}\label{eq:Fxy-c3-2}
2|\lambda|  \sqrt{n(\beta,\mu)} (y-\sqrt{n_0}) = 2|\lambda|  \sqrt{n(\beta,\mu)} \frac{y^2-n_0}{y+\sqrt{n_0}} \le c_3 |\lambda|^{2/3} (y^2-n_0)
\end{align}
for a constant $c_3 > 0$. Using \eqref{eq:Fxy-c3-1} and completing the square as in \eqref{eq:Fxy-c1-2}, we can bound  
\begin{align}\label{eq:Fxy-c3-3}
&c_2 (x-\sqrt{n_+})^2 +  \frac{\hat{v}(0)}{2}  (x^2+y^2)^2 - (\mu-\widetilde{\mu} + \eta^{-1} d_v)   (x^2+y^2) \nn\\
&\ge  -\frac{(\mu-\widetilde \mu + \eta^{-1}d_v)^2}{2\hat v(0)} +  \frac{\hat{v}(0)}{2} \left(x^2+y^2 - \frac{\mu-\widetilde{\mu}+\eta^{-1} d_v}{\hat{v}(0)}  \right)^2 +  c_2 (x^2-n_+)^2 \nn\\
&\ge  -\frac{(\mu-\widetilde \mu + \eta^{-1}d_v)^2}{2\hat v(0)}  + 2 c_4 \left( y^2 + n_+ - \frac{\mu-\widetilde{\mu}+\eta^{-1} d_v}{\hat{v}(0)} \right)^2\nn \\
&\ge  -\frac{(\mu-\widetilde \mu)^2}{2\hat v(0)}  + c_4 \left( y^2  - n_0 \right)^2  - C \eta^{-1}
\end{align}
with $c_4= \frac{1}{4}\min \left\{ \frac{\hat{v}(0)}{2} , c_2 \right\} > 0.$
Here we used the Cauchy--Schwarz inequality  $a^2+b^2\ge \frac{1}{2}(a+b)^2$ with $a=n_+-x^2$, $b=x^2+y^2- \frac{\mu-\widetilde{\mu}+\eta^{-1} d_v}{\hat{v}(0)}$. We also used \eqref{eq:rescaledvariablesminimization-b} in the form 
$$
n_+ - \frac{\mu-\widetilde{\mu}+\eta^{-1} d_v}{\hat{v}(0)} = -n_0 + O(\eta^{-1}). 
$$
From \eqref{eq:Fxy-c3-2} and \eqref{eq:Fxy-c3-3} we conclude that 
\begin{align}\label{eq:Fxy-c3-final}
  F^{\lambda,\delta}(x,y) & \ge  -\frac{(\mu-\widetilde \mu)^2}{2\hat v(0)}   + c_4 \left( y^2 - {n_0}\right)^2 + \delta y^2  - 2|\lambda| \sqrt{n(\beta,\mu) n_0} - c_3 |\lambda|^{2/3}(y^2-n_0)  - C \eta^{-1}
  \nn\\
  &\geq -\frac{(\mu-\widetilde \mu)^2}{2\hat v(0)} - 2 |\lambda| \sqrt{n(\beta,\mu) n_0} + \delta n_0 - C \eta^{-1} + c_4 (y^2-n_0)^2 - (|\delta| + c_3 |\lambda|^{2/3}) (y^2-n_0)  \nn\\
  &\ge -\frac{(\mu-\widetilde \mu)^2}{2\hat v(0)} - 2 |\lambda| \sqrt{n(\beta,\mu) n_0} + \delta n_0 - C ( \eta^{-1} + \delta^2 + |\lambda|^{4/3}).
  \end{align}

In summary, from \eqref{eq:Fxy-c0-final}, \eqref{eq:Fxy-c1-final} and \eqref{eq:Fxy-c3-final}, we find that for all $x,y\ge 0$, 
\begin{align}\label{eq:Fxy-all-final}
  F^{\lambda,\delta}(x,y)  \ge   -\frac{(\mu-\widetilde \mu)^2}{2\hat v(0)}  +  \delta  n_0  - 2|\lambda| \sqrt{n(\beta,\mu) n_0}  - C (\delta^2 + |\delta| |\lambda|^{1/3} + |\lambda| ^{4/3} + \eta^{-1} ). 
  \end{align}
 Inserting  \eqref{eq:Fxy-all-final} in  \eqref{eq:SSBCondensate6-b} and using $\int_{\mathbb{C}}  \zeta_{\Gamma}(z) \de z =1$,  we obtain
 \begin{align}\label{eq:SSB-Phi-delta-lambda-lower-con-proof}
\tr[ (\mathcal{H}_\eta^{\lambda,\delta} - \mu \mathcal{N}) \Gamma ]  - \frac{1}{\beta} S(\Gamma) &\geq - \frac{(\mu-\widetilde \mu)^2\eta}{2\hat v(0)}  +  \delta N_0(\beta,\widetilde \mu)  - 2|\lambda| \sqrt{N(\beta,\mu) N_0(\beta,\widetilde \mu)} \nn\\
&\quad - C \eta (\delta^2 + |\delta| |\lambda|^{1/3} + |\lambda| ^{4/3} + \eta^{-1/3} \ln \eta). 
\end{align}
Since this bound holds for all states $\Gamma$, we have the desired lower bound on the grand potential.  
\end{proof}


\section{Proof of the main result} \label{sec:proofmain}
In this section we shall prove Theorem \ref{thm:SSBCondensate}.  We will use a first-order
Griffith argument (i.e. a Hellmann–Feynman type argument), based on the estimate on the perturbed grand potential $\Phi^{\lambda,\delta}(\beta,\mu)$ in Proposition \ref{prop:grandpotperturbedupperlowerbound}. We shall divide the proof into three parts, each corresponding to the separate statements in \eqref{eq:thmmfbec}, \eqref{eq:SSBCondensate} and \eqref{eq:continuityCondensateFraction}.

\begin{proof}[Proof of Theorem \ref{thm:SSBCondensate} a)] Note that this part of the main theorem does not involve the symmetry breaking term.  Therefore, it suffices to consider the perturbed grand potential $\Phi^{\lambda,\delta}(\beta,\mu)$ in \eqref{eq:perturbedgrantPotential} with $\lambda=0$. We will first prove the second equality in \eqref{eq:thmmfbec}. To this end notice that \eqref{eq:perturbedgrantPotential} implies 
\begin{equation*}
\frac{\partial \Phi^{0,\delta}(\beta,\mu)}{\partial \delta}\Big|_{\delta=0} =   \tr[ a^*_0 a_0 G_{\beta,\mu} ].
\end{equation*}
On the other hand, by the Gibbs variational principle,   the map $\delta \mapsto \Phi^{0,\delta}(\beta,\mu)$ is concave (as an infimum over a family of affine functions). Hence
\begin{equation}\label{eq:lowerupperboundn0gibbs}
    \frac{\Phi^{0,\delta}(\beta,\mu)-\Phi^{0,0}(\beta,\mu)}{\delta}\leq \tr[ a^*_0 a_0 G_{\beta,\mu} ] \leq \frac{\Phi^{0,0}(\beta,\mu)-\Phi^{0,-\delta}(\beta,\mu)}{\delta}
\end{equation}
for any $\delta>0$. Using Proposition \ref{prop:grandpotperturbedupperlowerbound} with $\lambda=0$, we obtain
\begin{equation*}
\begin{aligned}
\left| \tr[ a^*_0 a_0 G_{\beta,\mu} ]  - N_0(\beta,\widetilde{\mu}) \right| \le \frac{C \eta (\delta^2 + \eta^{-1/3} \ln \eta)}{\delta}
\end{aligned}
\end{equation*}
for all $\delta>0$ and $\eta$ large. 
Choosing $\delta=\eta^{-\alpha}$ for some constant $\alpha \in (0,1/6)$ and  using Proposition \ref{prop:expnumberparticles} together with  \eqref{eq:condensateFractionEffectiveTheory}, we obtain the second equality in  \eqref{eq:thmmfbec}:
\begin{equation} \label{eq:abc}
\lim_{\eta\to \infty}    \frac{\tr[ a^*_0 a_0 G_{\beta,\mu} ] }{N(\beta,\mu)}=\lim_{\eta\to \infty}\frac{N_0(\beta,\widetilde{\mu})}{N(\beta,\mu)} = [1-\kappa^{-3/2}]_+. 
\end{equation}

In order to prove the first one, it is enough to show that for any $p\neq 0$ we have
\begin{equation}\label{eq:npbound}
      \tr[ a^*_p a_p G_{\beta,\mu} ] =o(\eta).
\end{equation}
This is true because $\tr[ a^*_p a_q G_{\beta,\mu} ] = 0$ for $p \neq q$, which follows from the translation-invariance of $G_{\beta,\mu}$. We will obtain this estimate using a Hellmann--Feynman argument. Let  $\epsilon\in (0,\frac12)$ and recall \eqref{eq:grandPotentialFunctional}. We have
$$\epsilon  \tr[ a^*_p a_p G_{\beta,\mu} ]=\mathcal{G}(G_{\beta,\mu})-\mathcal{G}_\epsilon(G_{\beta,\mu})$$
where
$$\mathcal{G}_\epsilon(G_{\beta,\mu})=\tr[(\mathcal{H}_{\eta} - \mu \mathcal{N} -\epsilon a_p^* a_p) G_{\beta,\mu}) ] - \frac{1}{\beta} S(G_{\beta,\mu})).$$
For an upper bound on $\mathcal{G}(G_{\beta,\mu})$ we simply use Proposition \ref{prop:grandpotupperbound}. To obtain a lower bound we repeat the proof of Proposition \ref{prop:grandpotlowerbound} in the case $\delta,\lambda=0$ with the only difference that now the ideal gas corresponds to the one-body Hamiltonian $d\Upsilon(Q(-\Delta)-\epsilon |p\rangle \langle p|).$ In particular, we obtain the lower bound of the form 
$$\mathcal{G}_\epsilon(G_{\beta,\mu})\geq \Phi_+^{\mathrm{id},\epsilon} (\beta,\widetilde{\mu}) - \frac{(\mu-\widetilde{\mu})^2\eta}{2\hat{v}(0)} - C \eta^{2/3} \ln \eta, $$
where 
$$\Phi_+^{\mathrm{id},\epsilon} (\beta,\widetilde{\mu})=   \frac{1}{\beta} \sum_{k \in \Lambda_+^*, k\neq p} \ln\left( 1 - \exp(-\beta(k^2 - \widetilde{\mu})) \right)+\ln\left( 1 - \exp(-\beta(p^2-\epsilon - \widetilde{\mu})) \right).$$
Applying this estimate we obtain
\begin{equation} \label{eq:epsilonnpbound}
\epsilon  \tr[ a^*_p a_p G_{\beta,\mu} ]\leq   \frac{1}{\beta} \left[\ln\left( 1 - \exp(-\beta(p^2 - \widetilde{\mu})) \right)-\ln\left( 1 - \exp(-\beta(p^2-\epsilon - \widetilde{\mu})) \right)\right]+ C \eta^{2/3} \ln \eta.
\end{equation}
Consider the function $f: (0,\frac12)\to \mathbb{R}$ given by 
$$f(\epsilon)=\ln\left( 1 - \exp(-\beta(p^2-\epsilon - \widetilde{\mu})) \right).$$
Since 
$$f'(\epsilon)=-\frac{\beta}{\exp(-\beta(p^2-\epsilon - \widetilde{\mu}))-1},$$
using $\widetilde{\mu}<0$ we obtain
$$|f'(\epsilon)|\leq C,$$
which applied to \eqref{eq:epsilonnpbound} implies
$$\epsilon  \tr[ a^*_p a_p G_{\beta,\mu} ]\leq C \eta \ln \eta. $$ 
Choosing, e.g., $\epsilon=\frac14$ yields \eqref{eq:npbound}. This ends the proof of part a) of Theorem \ref{thm:SSBCondensate}.
\end{proof}

Due to technical reasons that soon become clear, let us first prove part c) of Theorem \ref{thm:SSBCondensate} and then turn to  part b).  
\begin{proof}[Proof of Theorem \ref{thm:SSBCondensate} c)]
The proof follows the same strategy as the in the proof of the second equality in \eqref{eq:thmmfbec}, but this time we keep $\lambda \neq 0$. Using the same argument as for   \eqref{eq:lowerupperboundn0gibbs} we obtain for any $\delta>0$
\begin{equation}\label{eq:lowerupperboundn0gibbslambda}
    \frac{\Phi^{\lambda,\delta}(\beta,\mu)-\Phi^{\lambda,0}(\beta,\mu)}{\delta}\leq \tr[ a^*_0 a_0 G^{\lambda}_{\beta,\mu} ] \leq \frac{\Phi^{\lambda,0}(\beta,\mu)-\Phi^{\lambda,-\delta}(\beta,\mu)}{\delta}.
\end{equation}
Therefore, Proposition \ref{prop:grandpotperturbedupperlowerbound} implies that 
$$
\left| \tr[ a^*_0 a_0 G^{\lambda}_{\beta,\mu} ] - N_0(\beta,\widetilde \mu) \right| \le \frac{C \eta (\delta^2 + |\delta| |\lambda|^{1/3} + |\lambda| ^{4/3} + \eta^{-1/3} \ln \eta)}{\delta}
$$
for all $\delta>0$, $\lambda\in \R$ and $\eta$ large enough. Taking $\delta =|\lambda|$ and using Proposition \ref{prop:expnumberparticles} and \eqref{eq:condensateFractionEffectiveTheory}, we obtain \eqref{eq:continuityCondensateFraction}:
\begin{equation} \label{eq:abc-2}
\lim_{\lambda\to 0}\lim_{\eta\to \infty}    \frac{\tr[ a^*_0 a_0 G_{\beta,\mu} ] }{N(\beta,\mu)}=\lim_{\lambda\to 0}\lim_{\eta\to \infty} \left( \frac{N_0(\beta,\widetilde{\mu})}{N(\beta,\mu)} + O(|\lambda|^{1/3}) + O(|\lambda|^{-1} \eta^{-1/3} \ln \eta) \right) = [1-\kappa^{-3/2}]_+. 
\end{equation}
\end{proof}

It remains to prove \eqref{eq:SSBCondensate}. 

\begin{proof}[Proof of Theorem \ref{thm:SSBCondensate} b)] Note that by  the Cauchy--Schwarz  inequality, the result in part c) implies the result in part b) if $\kappa\le 1$. Therefore, it remains to consider the condensed phase $\kappa>1$.  Now we consider the perturbed grand potential $\Phi^{\lambda,\delta}(\beta,\mu)$ in \eqref{eq:perturbedgrantPotential} with $\delta=0$.  We proceed as before and obtain 
\begin{equation*}
\frac{\partial \Phi^{\lambda,0}(\beta,\mu)}{\partial \lambda} =  \sqrt{N(\beta,\mu)} \tr[ (a^*_0 +a_0) G^\lambda_{\beta,\mu} ],
\end{equation*}
which, by concavity of $\Phi^{\lambda,0}(\beta,\mu)$, implies
\begin{equation}
\frac{\Phi^{\lambda+\epsilon,0}(\beta,\mu)-\Phi^{\lambda,0}(\beta,\mu)}{\epsilon}\leq \sqrt{N(\beta,\mu)} \tr[ (a^*_0 +a_0) G^\lambda_{\beta,\mu} ]\leq \frac{\Phi^{\lambda,0}(\beta,\mu)-\Phi^{\lambda-\epsilon,0}(\beta,\mu)}{\epsilon}.
\end{equation}
Using Proposition \ref{prop:grandpotperturbedupperlowerbound} we obtain 
\begin{equation*}
\left| \sqrt{N(\beta,\mu)} \tr[ (a^*_0 +a_0) G^\lambda_{\beta,\mu} ] - 2 \sqrt{N(\beta,\mu) N_0(\beta,\widetilde \mu)} \right| \le \frac{C \eta ( |\lambda| ^{4/3} + \eta^{-1/3} \ln \eta)}{\eps}. 
\end{equation*}
Choosing $\epsilon=|\lambda|/2$ and using Proposition \ref{prop:expnumberparticles} and \eqref{eq:condensateFractionEffectiveTheory} yields 
\begin{equation}
    \lim_{\lambda \to 0} \lim_{N \to \infty} \frac{|\tr[ (a_0^*+a_0) G^{\lambda}_{\beta,\mu} ]|}{2 N(\beta,\mu)^{1/2}} = \sqrt{\left[ 1 - \frac{1}{\kappa^{3/2}} \right]_+}.
    \label{eq:SSBCondensatereal}
\end{equation}

Finally, we shall use \eqref{eq:SSBCondensatereal} and \eqref{eq:continuityCondensateFraction} to prove \eqref{eq:SSBCondensate}. From \eqref{eq:SSBCondensatereal} we have  
\begin{equation} \label{eq:realpartwlimit}
\lim_{\lambda\to 0}\lim_{\eta\to \infty}|\Re(w(\lambda,\eta))|=  \sqrt{\left[ 1 - \frac{1}{\kappa^{3/2}} \right]_+},\quad \text{with}\quad w(\lambda,\eta)=\frac{\tr[ a_0 G^{\lambda}_{\beta,  \mu} ]}{\sqrt{N(\beta,\mu)}}. 
\end{equation}
On the other hand, by the Cauchy-Schwartz inequality, we have
$$|w(\lambda,\eta)|^2\leq \frac{|\tr[ a_0^* a_0 G^{\lambda}_{\beta,   \mu} ]|}{  N(\beta,\mu)}$$
which by \eqref{eq:continuityCondensateFraction} implies
\begin{equation} \label{eq:wsquaredlimit}
    \lim_{\lambda\to 0}\lim_{\eta\to \infty}|w(\lambda,\eta)|^2\leq \left[ 1 - \frac{1}{\kappa^{3/2}} \right]_+. 
\end{equation}
But \eqref{eq:realpartwlimit} and \eqref{eq:wsquaredlimit} imply that
\begin{equation*} 
\lim_{\lambda\to 0}\lim_{\eta\to \infty}|\Im (w(\lambda,\eta))|=0.
\end{equation*}
In combination, this result and \eqref{eq:SSBCondensatereal} prove \eqref{eq:SSBCondensate}.
\end{proof}

\appendix
\section{Properties of the effective chemical potential}
\label{app:effectiveChemicalPotential}
In this section we investigate equation \eqref{eq:GrantCanonicalEffectiveIddealGasChemPot} for the effective chemical potential appearing in our statements for the grand potential. The first lemma guarantees the existence of a unique solution and provides a priori bounds.
\begin{lemma}\label{lem:effectiveChemicalPotentialAPriori}
    The following three statements hold:
    \begin{enumerate}[label=(\alph*)]
    \item Assume that $\beta,\eta, \hat{v}(0) > 0$ and $\mu \in \mathbb{R}$. The equation
    \begin{equation}
        \sum_{p \in \Lambda^*} \frac{1}{e^{\beta(p^2 - \widetilde{\mu})}-1} = \frac{(\mu - \widetilde{\mu})\eta}{\hat{v}(0)}
        \label{eq:GrantCanonicalEffectiveIddealGasChemPotv2}
    \end{equation}
    for $\widetilde{\mu}$ admits a unique solution in the set $(-\infty,0)$.  
    \item We consider the limit $\eta \to \infty$, $\beta = \kappa \beta_{\mathrm{c}}(\eta)$ with $\kappa \in (0,\infty)$ and $\beta_{\mathrm{c}}$ in \eqref{eq:crittemp}. We assume that $\mu$, which may depend on $\eta$, satisfies $-\eta^{2/3} \lesssim \mu \lesssim 1$. There exists a constant $c>0$ such that the unique solution to \eqref{eq:GrantCanonicalEffectiveIddealGasChemPotv2} satisfies 
    \begin{equation}
        c \leq \mu - \widetilde{\mu} \leq c^{-1}.
        \label{eq:muMinusMu0Bound}
    \end{equation}
    Moreover, if $\mu \geq 0$ then $-\widetilde{\mu} \lesssim 1$ and if $\mu < 0$ we have $-\widetilde{\mu} \lesssim \eta^{2/3}$. 
    \item Under the assumptions stated in part~(b) there exists a constant $c > 0$ such that
    \begin{equation}
        c \eta \leq \sum_{p \in \Lambda^*} \frac{1}{e^{\beta(p^2 - \widetilde{\mu})}-1} \leq c^{-1} \eta.
        \label{eq:appAndi1}
    \end{equation}
    \end{enumerate}
\end{lemma}

Before we prove the above lemma, we state and prove a lemma that allows us to approximate momentum sums by integrals.  
\begin{proof}
    For $\widetilde{\mu} \in (-\infty,0)$ we define the function
    \begin{equation}
        f(\widetilde{\mu}) = \sum_{p \in \Lambda^*} \frac{1}{e^{\beta(p^2 - \widetilde{\mu})}-1} +  \frac{(\widetilde{\mu} - \mu)\eta}{\hat{v}(0)}.
        \label{eq:EffectiveChemicalPOtential1}
    \end{equation}
    It is not difficult to check that $f$ is continuous, strictly monotone increasing, and satisfies $f(\widetilde{\mu}) \to -\infty$ for $\widetilde{\mu} \to - \infty$ and $f(\widetilde{\mu}) \to +\infty$ for $\widetilde{\mu} \to 0$. This implies part~(a) and it remains to prove parts~(b) and (c).

    With \eqref{eq:GrantCanonicalEffectiveIddealGasChemPotv2}, $\widetilde{\mu} < 0$, $(\exp(x)-1)^{-1} \leq x$ for $x>0$, and an argument based on a Riemann sum approximation, it is not difficult to see that
    \begin{equation}
        \mu - \widetilde{\mu} \lesssim \frac{1}{\eta} \left[ \frac{1}{-\beta \widetilde{\mu}} + \beta^{-3/2} \int_{\mathbb{R}^3} \frac{1}{\exp(p^2)-1} \de p  \right] \lesssim \frac{1}{-\eta^{1/3} \widetilde{\mu}} + 1.
        \label{eq:EffectiveChemicalPOtential2}    
    \end{equation}
    If $\mu \geq 0$ the above inequality implies $-\widetilde{\mu} \lesssim 1$. In the case $\mu < 0$, we additionally use $\mu \gesssim - \eta^{2/3}$ and find $-\widetilde{\mu} \lesssim \eta^{2/3}$.

    We again interpret the sum over momenta as a Riemann sum and use $-\widetilde{\mu} \lesssim \eta^{2/3}$ to check that
    \begin{equation}
        \mu - \widetilde{\mu} \gesssim \frac{1}{\eta \beta^{3/2}}  \int_{\mathbb{R}^3} \frac{1}{\exp(p^2 + c)-1} \de p \gesssim 1.
        \label{eq:EffectiveChemicalPOtential3} 
    \end{equation}
    This proves the lower bound in \eqref{eq:muMinusMu0Bound}. To derive an upper bound for $\mu - \widetilde{\mu}$ in the case $\mu < 0$, we combine \eqref{eq:EffectiveChemicalPOtential3} and \eqref{eq:EffectiveChemicalPOtential2} as follows:
    \begin{equation}
        \mu - \widetilde{\mu} \lesssim \frac{1}{-\eta^{1/3} \widetilde{\mu}} + 1 \lesssim \frac{1}{\eta^{1/3}(1-\mu)} + 1 \leq \frac{1}{\eta^{1/3}} + 1 \lesssim 1.
        \label{eq:EffectiveChemicalPOtential4} 
    \end{equation}
    This proves part~(b).

    Part~(c) of Lemma~\ref{lem:effectiveChemicalPotentialAPriori} follows from \eqref{eq:GrantCanonicalEffectiveIddealGasChemPotv2} and \eqref{eq:muMinusMu0Bound}. 
    \end{proof}

 \vspace{0.5cm}

\textbf{Acknowledgments.} A. D. gratefully acknowledges funding from the Swiss National Science Foundation (SNSF) through the Ambizione grant PZ00P2 185851. P. T. N. was partially supported by the European Research Council via the ERC Consolidator Grant RAMBAS (Project-Nr. 10104424). The work of M. N. was supported by the Polish-German NCN-DFG
grant Beethoven Classic 3 (project no. 2018/31/G/ST1/01166). 
\vspace{0.5cm}

\textbf{Conflict of interest statement.} On behalf of all authors, the corresponding author states that there is no conflict of interest.

\vspace{0.5cm}

\textbf{Data availability statement.} No new data were created or analysed in this study. Data sharing is not applicable to this article.

\vspace{0.5cm}

\bibliographystyle{siam}

\vspace{0.5cm}

\noindent (Andreas Deuchert) Department of Mathematics, Virginia Tech \\ 225 Stanger Street, Blacksburg, VA, 24060-1026, USA\\
E-mail address: \texttt{andreas.deuchert@vt.edu}
\vspace{0.5cm}

\noindent(Phan Th\`anh Nam) Department of Mathematics, LMU Munich, Theresienstrasse 39, 80333 Munich, Germany\\
E-mail address: \texttt{nam@math.lmu.de}
\vspace{0.5cm}

\noindent (Marcin Napi\'orkowski) Department of Mathematical Methods in Physics, Faculty of Physics, University of Warsaw, Pasteura 5, 02-093, Warsaw, Poland\\
E-mail address: \texttt{marcin.napiorkowski@fuw.edu.pl}

\end{document}